\documentclass{siamart190516}

\usepackage{bm}

\usepackage{enumerate}
\usepackage[top=1in, bottom=1in, left=1in, right=1in]{geometry}
\usepackage{mathrsfs,color}
\usepackage{amsfonts}\usepackage{multirow}
\usepackage{amssymb}
\usepackage{dsfont}
\usepackage{todonotes}
\usepackage{caption,comment}
\usepackage{blkarray}
\usepackage{amsmath}
\usepackage{algpseudocode}
\usepackage{float}
\usepackage[toc,page]{appendix}
\usepackage{graphicx,afterpage}
\usepackage{epstopdf}
\usepackage{cancel}

\usepackage{mathdots}

\usepackage{todonotes}

\newcommand{\Uhat}{\widehat{U}}

\newcommand{\Xhat}{\widehat{X}}

\newcommand{\Vhat}{\widehat{V}}

\newcommand{\Mhat}{\widehat{M}}
\newcommand{\Qhat}{\widehat{Q}}
\newcommand{\rhohat}{\widehat{\rho}}
\newcommand{\alphahat}{\widehat{\alpha}}
\newcommand{\betahat}{\widehat{\beta}}
\newcommand{\kappahat}{\widehat{\kappa}}

\newcommand{\E}{\mathbb{E}}

\newcommand{\reals}{\mathbb{R}}

\newcommand{\nuhat}{\widehat{\nu}}

\newcommand{\R}{\mathbb{R}}

\newcommand{\Rhat}{\widehat{R}}

\newcommand{\Shat}{\widehat{S}}

\newcommand{\pihat}{\widehat{\pi}}
\newcommand{\Phat}{\widehat{P}}
\newcommand{\Fhat}{\widehat{F}}

\newcommand{\qhat}{\hat{q}}

\newcommand{\var}{\text{var}}

\newcommand{\fhat}{\hat{f}}

\numberwithin{equation}{section}
\renewcommand{\theequation}{\thesection .\arabic{equation}}

\newsiamremark{remark}{Remark}
\newsiamremark{hypothesis}{Hypothesis}
\crefname{hypothesis}{Hypothesis}{Hypotheses}
\newsiamthm{claim}{Claim}
\newsiamthm{aspt}{Assumption}
\newsiamthm{lem}{Lemma}
\newsiamthm{prop}{Proposition}
\newsiamthm{thm}{Theorem}

\headers{MCMC under Numerical Perturbation}{T. Cui, J. Dong, A. Jasra and  X. T. Tong}

\title{Convergence Speed and Approximation Accuracy of Numerical MCMC\thanks{Submitted to the editors DATE.
\funding{TC is supported by the Australian Research Council  grant DP210103092. AJ is supported by KAUST baseline funding. XT is supported by the Singapore Ministry of Education (MOE) grant R-146-000-292-114. }}}

\author{Tiangang Cui\thanks{School of Mathematics, Monash University, Australia
  (\email{tiangang.cui@monash.edu}).}
\and Jing Dong\thanks{Graduate School of Business, Columbia University,  U.S.A.
  (\email{jing.dong@gsb.columbia.edu}).}
  \and Ajay Jasra\thanks{Applied Mathematics and Computational Science Program, Computer, Electrical, Mathematical Sciences and Engineering Division, King Abdullah University of Science and Technology,  KSA
  (\email{ajay.jasra@kaust.edu.sa}).}
  \and Xin T. Tong\thanks{Department of Mathematics, National University of Singapore, Singapore 
  (\email{mattxin@nus.edu.sg}).}
}

\usepackage{amsopn}

\ifpdf
\hypersetup{
  pdftitle={},
  pdfauthor={T. Cui, J. Dong, A. Jasra and X. T. Tong,}
}
\fi

\begin{document}

\maketitle

\begin{abstract}
When implementing Markov Chain Monte Carlo (MCMC) algorithms, perturbation caused by numerical errors is sometimes inevitable.
This paper studies how perturbation of MCMC affects the convergence speed and Monte Carlo estimation accuracy.
Our results show that when the original Markov chain converges to stationarity fast enough and the perturbed transition kernel is a good approximation to the original transition kernel, the corresponding perturbed sampler has similar convergence speed and high approximation accuracy as well.
We discuss two different analysis frameworks: ergodicity and spectral gap, both are widely used in the literature. 
Our results can be easily extended to obtain non-asymptotic error bounds for MCMC estimators. 
 We also demonstrate how to apply our convergence and approximation results to the analysis of specific sampling algorithms, including Random walk Metropolis and Metropolis adjusted Langevin algorithm with perturbed target densities, and parallel tempering Monte Carlo with perturbed densities. Finally we present some simple numerical examples to verify our theoretical claims.
\end{abstract}



%
%
%


%

\begin{abstract}
When implementing Markov Chain Monte Carlo (MCMC) algorithms, perturbation caused by numerical errors is sometimes inevitable.
This paper studies how perturbation of MCMC affect the convergence speed and approximation accuracy.
Our results show that when the original Markov chain converges to stationarity fast enough and the perturbed transition kernel is a good approximation to the original transition kernel, the corresponding perturbed sampler has fast convergence speed and high approximation accuracy as well.
We discuss two different analysis frameworks: ergodicity and spectral gap, both are widely used in the literature. 
Our results can be easily extended to obtain non-asymptotic error bounds for MCMC estimators. 
 We also demonstrate how to apply our convergence and approximation results to the analysis of specific sampling algorithms, including Random walk Metropolis and Metropolis adjusted Langevin algorithm with perturbed target densities, and parallel tempering Monte Carlo with perturbed densities. Finally we present some simple numerical examples to verify our theoretical claims.
\end{abstract}

\section{Introduction}
Markov Chain Monte Carlo (MCMC) is one of the main sampling methods in Bayesian statistics.
Given a target density $\pi$ on $\reals^d$, it simulates a Markov chain $X_n$ with  transition kernel $P$, such that  $\pi$ is the corresponding invariant measure.
Under some generic conditions, the distribution of $X_n$ converges to $\pi$ geometrically fast.
This indicates the existence of some mixing time $n_0$, so that the distribution of $X_n$ is close to $\pi$ when $n>n_0$.
In other words, if we can use the following approximation
\begin{equation}
\label{eq:onesample}
\E f(X_{n})\approx \E^\pi f(X):=\int f(x)\pi(x)dx.
\end{equation}
In practice, this allows us to approximate the average of a test function  $\E^\pi f(X)$ using the temporal average of the Markov chain:
\begin{equation}
\label{eq:tempavg}
F_n:=\frac{1}{n}\sum_{i=1}^n f(X_{n_0+i}).
\end{equation}
The efficiency of the approximation scheme  \eqref{eq:tempavg} is largely determined by the  convergence speed of the Markov Chain $X_n$ to $\pi$.  
In particular, the ``burning" sample size $n_0$ should be set such that the distribution at step $n_0$ is close to $\pi$  and the effective sample size of $F_n$ is approximately $O(n/n_0)$.
In this context, convergence analysis has been a key component in the MCMC literature 
(see, for example, Section 4.1 of \cite{andrieu2003introduction}).
%

When implementing MCMC on complicated target densities, it is often the case that we can only simulate a perturbed Markov chain $\Xhat_n$ with transition kernel $\Phat$. This is mainly due to two reasons:
\begin{enumerate}
\item The transition kernel $P$ cannot be simulated directly. For example, if $X_n$ is described by a stochastic differential equation (SDE), using numerical simulation method like the Euler-Maruyama method will induce discretization errors. 
\item We do not have direct access to $\pi(x)$ or even an un-normalized version of it. This is quite common in Bayesian inverse problems \cite{stuart2010inverse}, where the target density can be written as 
\begin{equation}
\label{eq:pi1}
\pi(x)\propto p_0(x)\exp\left(-\tfrac12 \|G(x)-y\|^2\right).
\end{equation}
In \eqref{eq:pi1}, $p_0$ is the prior density of the unknown parameter $x$, $G$ describes the data generating process, and $y$ is the collected data. In many cases, $G$ is formulated through an involved partial differential equation, and we can only compute an approximation of it, $\widehat{G}$, numerically \cite{bui2013computational, hairer2014spectral, beskos2018multilevel}. The corresponding ``numerical" density becomes 
\begin{equation}
\label{eq:pi}
\pihat(x)\propto p_0(x)\exp\left(-\tfrac12 \|\widehat{G}(x)-y\|^2\right).
\end{equation}
\end{enumerate}
In the above-mentioned cases, we run an MCMC $\Xhat_n$ with transition kernel $\Phat$ and target density $\pihat$, which is the invariant measure of $\Xhat_n$,  instead of $\pi$. 
One important difference between the two scenarios is that we often know $\pihat$ explicitly in the second scenario, but not in the first one. 
Following \eqref{eq:onesample} and \eqref{eq:tempavg},  in both scenarios listed above, we would like to approximate $\E^{\pi} f(X)$ using 
\begin{equation}
\label{eq:numMC}
\E f(\Xhat_{n_0})\quad \text{or}\quad \Fhat_n=\frac{1}{n}\sum_{i=1}^n f(\Xhat_{n_0+i}).
\end{equation}

%
%

There are two key questions to address when using estimators of form \eqref{eq:numMC}.  The first question is about the convergence speed of $\Xhat_n$ towards its invariant measure $\pihat$, which determines the efficiency of the estimators in \eqref{eq:tempavg}. In particular, if we use
 $D$ to denote some metric between two distributions and $\nu$ to denote the distribution of $\Xhat_0$, we are interested in how fast  $D(\nu\Phat^n, \pihat)$ converges to zero. The second question is about approximation accuracy, which can be measured by either the distance between the two invariant measures, $D(\pihat, \pi)$, or the distance between the distribution of $\Xhat_n$ and $\pi$, $D(\nu\Phat^n, \pi)$. 

%
%
For $\Xhat_n$ to achieve fast convergence and high approximation accuracy, we need to impose the following two high-level conditions (these conditions will be made more precise in our subsequent development):  
\begin{enumerate}[1.]
\item $\Phat$ is a good approximation of $P$. 
\item $X_n$ converges to its invariant measure $\pi$ fast enough.
\end{enumerate}
Condition 1 is necessary, because if $\Phat$ is not a good approximation of $P$, $\pihat$ is unlikely to be close to $\pi$, and the  convergence property of $X_n$ will not be useful in inferring the convergence property of $\Phat$. 
Condition 2 is also necessary. Otherwise, the approximation error may grow exponentially with the number of iterations. Since Condition 1 involves only the one-step transition kernels,  it is easier to fulfill. Hence, it is reasonable to study Condition 2 first and then formulate a version Condition 1 that is compatible to the corresponding Condition 2.

In the literature of Markov processes, Condition 2 is often studied using one of the two frameworks: the ergodicity framework and the spectral gap framework. 
The main differences between these two frameworks are the metrics involved and analysis tools involved.  
The ergodicity framework measures the convergence rate of $\nu P^n$ to $\pi$ in the total variation distance or more general Wasserstein metrics \cite{meyn2012markov}.
Establishing ergodicity often involves finding an appropriate Lyapunov function and constructing an appropriate coupling \cite{lindvall02}. 
The spectral gap framework measures the convergence rate of $\nu P^n$ to $\pi$ in $\chi^2$-distance \cite{bakry2014analysis} or KL-divergence\cite{vempala2019rapid}. 
Bounding the spectral gap often requires functional analysis or other partial differential equation (PDE) tools such as Poincar\'e inequality and log Sobolev inequality. We also note that these two frameworks are related. In particular, on one hand, under suitable regularity conditions, ergodicity leads to the existence of a spectral gap (see, e.g., Proposition 2.8 in \cite{hairer2014spectral}). On the other hand, under proper regularity conditions, convergence under the $\chi^2$ distance leads to convergence under the total variation distance.

We discuss both frameworks in this paper, because for some Markov processes, we may only have knowledge of one form of convergence. For example, to the best of our knowledge, the parallel tempering methods are only studied under the spectral gap framework \cite{woodard2009conditions,dong2020gap}. Preconditioned Crank--Nicolson 
algorithm is only studied under the ergodicity framework \cite{hairer2014spectral} (Although the paper's title starts with ``spectral gap", it is more in line with the ergodicity framework described above).
The unadjusted Langevin algorithm was first studied under the ergodicity framework \cite{durmus2017nonasymptotic,JMLR:v20:19-306} and later under the spectral gap framework \cite{vempala2019rapid}. While there might be theoretical value to establish convergence in both frameworks, this is practically unnecessary.  In this paper, we assume the convergence of $X_n$ under either the ergodicity framework or the spectral gap framework and study the convergence of $\Xhat_n$ under one of the two frameworks accordingly. We not only address the question qualitatively, but also quantitatively by establishing bounds for the convergence speed of $\Xhat_n$ with respect to the convergence speed of $X_n$ and the approximation accuracy.

\subsection{Related literature}
The approximation and convergence questions we study here are fundamental for MCMC and have been studied in various settings before. 
Most existing works focus on specific approximation schemes. 
For example, 
\cite{herve2014approximating} studies the ergodicity property of finite-rank non-negative sub-Markov kernels 
in relation to the ergodicity property of the original Markov kernel.
\cite{bardenet2014adaptive} studies the convergence and approximation problems of an adaptive subsampling approach under the assumption of uniform ergodicity. \cite{medina2020perturbation} studies the approximation problem for Monte Carlo within Markov Chain algorithms.
\cite{JMLR:v20:19-306} studies the approximation problem for several sampling algorithms when the target distribution is log-concave.
Overall, there is a lack of a unified framework.

To the best of our knowledge, there are only two papers that provide a general discussion similar to ours \cite{shardlow2000perturbation, rudolf2018perturbation}, but these two papers focus on the approximation accuracy under the ergodicity framework. 
How to quantify the convergence speed and approximation accuracy of $\Xhat_n$ under the spectral gap framework is largely missing in the literature.
Moreover, while the connection between ergodicity and MCMC sampling error is well known, most results are asymptotic, i.e., in the form of central limit theorems \cite{jones2004markov}. Non-asymptotic error bounds are more useful in practice \cite{joulin2010curvature}. Our work intends to fill these gaps and provides a complete list of performance quantifications for  numerical MCMC samplers (perturbed Markov processes). We also demonstrate that our results can be easily applied to the analysis of various algorithms in Sections \ref{sec:Metro} and \ref{sec:RE}.

\subsection{Notations}
Let $\Omega$ denote a Polish space and $\mathcal{B}(\Omega)$ denote the corresponding Borel $\sigma$-algebra.
For a probability measure $\mu$ on $\Omega$, we define 
\[\mu f=\int_{\Omega} f(x)\mu(dx),\quad \var_\mu f=\int_{\Omega} (f(x)-\mu f)^2\mu(dx).
\]
We also use $\mu(x)$ to denote the corresponding density function.
For measurable functions $f,g: \Omega \rightarrow \mathbb{R}$, we define the inner product with respect to $\mu$ as
\[\langle f,g\rangle_{\mu}=\int_{\Omega} f(x)g(x)\mu(dx).\]
Then, $\|f\|_{\mu}^2 =\langle f,f\rangle_{\mu}=\int_{\Omega} f(x)^2\mu(dx)$.
In what follows, we omit $\Omega$ from the integral notation when it is clear from the context.
For a transition kernel $P$, define
\[\mu P(A)=\int P(x,A)\mu(dx).\]
For a measurable function $f$, we also define
$\delta_xPf=Pf(x)=\int f(y)P(x,dy)$.
Suppose $P$ is irreducible and symmetric with respect to $\pi$, 
then for any measurable functions $f,g$,
\[
\langle Pf,g\rangle_{\pi} =\langle f, Pg\rangle_{\pi}.
\]
%
Lastly, we denote $C$ as a generic constant whose value can change from line to line.

\subsection{Organization}
We start by developing general analysis results for the ergodicity framework in Section  \ref{sec:ergodicity}, and the spectral gap framework in Section \ref{sec:gap}.
We demonstrate how to apply these frameworks on two popular MH-MCMCs in Section \ref{sec:Metro}, and on the involved parallel tempering algorithm in Section \ref{sec:RE}. 
Finally in Section \ref{sec:num}, we verify our claims numerically on an Bayesian inverse problem, which tries to infer initial condition and model parameter in the predator-prey system.

\section{The Ergodicity Framework} \label{sec:ergodicity}
We start our discussion with the ergodicity framework. 
Following \cite{rudolf2018perturbation}, we first introduce the metric we use and the notion of ergodicity.
For a measurable function $V:\reals^d \rightarrow [1,\infty]$, define
\[
d_V(x,y)=(V(x)+V(y))1_{x\neq y}.
\]
For two probability measures $\mu$ and $\nu$ on $\reals^d$, define
\[
\|\mu-\nu\|_V=\sup_{|f|\leq V}\left |\int f(x)(\mu(dx)-\nu(dx))\right |.
\]
It can be shown that $\|\mu-\nu\|_V=W_{d_V}(\mu,\nu)$ where $W$ denote the Wasserstein distance (Lemma 3.1 in \cite{rudolf2018perturbation}).
If we use the constant function $V(x)=1$, this gives the  well known total variation distance, i.e.,
\[\|\mu-\nu\|_{TV} = \sup_{|f|\leq 1}\left |\int f(x)(\mu(dx)-\nu(dx))\right |.\]
Note that using $V(x)=1$ neglects the location information of $x$. This location information can be crucial for problems with unbounded domain. 
In general, for problems with unbounded domain, one often chooses $V$ to be a Lyapunov function. Given a  Markov chain $(X_n,P)$, we say $V:\reals^d\to [1,\infty)$ is a Lyapunov function if there exist $\lambda\in(0,1)$ and $L>0$, such that 
\begin{equation}
\label{eq:lyap}
PV (x)=\int P(x,dy) V(y) \leq \lambda V(x) + L.
\end{equation}

If $\pi$ is the invariant measure of $X_n$, we say $X_n$ is geometrically ergodic under $d_V$  (see Theorem 16.1 in \cite{meyn2012markov}) if there are constants $\rho\in(0,1)$ and $C_0\in(0,\infty)$, such that
for any $n\in \mathbb{Z}^+$,
\begin{equation} \label{eq:ergo1}
\|\delta_x P^n - \pi\|_V \leq C_0 \rho^n V(x).
\end{equation}
We refer to $\rho$ as the ergodicity coefficient.
 Note that the smaller the value of $\rho$, the faster the convergence to stationarity.
From \eqref{eq:ergo1}, using the triangle inequality we obtain an equivalent definition of geometric ergodicity, which requires that for any $x$ and $y$
\begin{equation} \label{eq:ergo}
\|\delta_x P^n - \delta_yP^n\|_V \leq C_0^{\prime}\rho^n d_V(x,y). 
\end{equation}
The equivalence can be seen from
\begin{equation}
\label{tmp:ergo}
\|\delta_x P^n - \pi\|_V\leq \int \pi(dy)\|\delta_x P^n - \delta_yP^n\|_V \leq C_0^{\prime}\rho^n(V(x)+ \pi  V)\leq C_0\rho^nV(x).
\end{equation}
where $C_0=C_0^{\prime}(1+ \pi V )$.

%

The approximation problem under the ergodicity framework has been studied in \cite{rudolf2018perturbation}.
We present one of their main results here which is related to our subsequent development. 

\begin{theorem}[Theorem 3.1 in \cite{rudolf2018perturbation}]
\label{thm:rs}
Suppose $(X_n, P)$ is geometrically ergodic, i.e., as in \eqref{eq:ergo}. Suppose $\Vhat$ is a Lyapunov function for $(\Xhat_n,\Phat)$ in the sense of \eqref{eq:lyap}, and
\begin{equation}\label{eq:epsilon_close1}
\|\delta_x P - \delta_x\Phat \|_V \leq \epsilon \Vhat(x).
\end{equation}
Then,  for some constant $C$, we have 
\begin{equation}\label{eq:epsilon_close1n}
\|\delta_x P^n - \delta_x\Phat^n\|_V \leq C\epsilon \frac{1-\rho^n}{1-\rho}\left(\Vhat(x) + \frac{L}{1-\lambda}\right).
\end{equation}
\end{theorem}
The bound in \eqref{eq:epsilon_close1n} and the triangular inequality give us an approximation error bound 
\[
\|\delta_x \Phat^n - \pi\|_V\leq \|\delta_x P^n - \delta_x\Phat^n\|_V+\|\delta_x P^n-\pi\|_V\leq C'(\epsilon+\rho^n) (\Vhat(x)+V(x))
\]
for some constant $C'$. 

Note that the right hand side of \eqref{eq:epsilon_close1n} is not converging to zero as  $n\to 0$. Thus, it cannot help us learn the ergodicity of $\Xhat_n$ or whether $\Xhat_n$ has a unique invariant measure. 
The next result shows  that ergodicity can be obtained with essentially the same conditions as Theorem \ref{thm:rs} (note that condition \eqref{eq:ergo1} leads to  \eqref{eq:close1} through \eqref{tmp:ergo}).
%
%

\begin{theorem} \label{th:ergo}
Suppose $V$ is a Lypaunov function for $P$ in the sense of \eqref{eq:lyap}.
In addition, assume there exist $N\in \mathbb{Z}^+$ and $\rho\in(0,1)$, such that for any $n\geq N$,
\begin{equation}
\label{eq:close1}
\|\delta_xP^n-\delta_y P^n\|_V\leq \rho^nd_V(x,y).
\end{equation}
Lastly, suppose the following holds for a sufficiently small $\epsilon>0$,
\begin{equation}\label{eq:close2}
\|\delta_x P - \delta_x \Phat\|_V\leq \epsilon V (x).
\end{equation}
Then, $V$ is a Lyapunov function for $\Phat$ as well with 
\[\Phat V(x)\leq (\lambda+\epsilon) V(x) + L.\]
Moreover, $\Xhat_n$ has a unique invariant measure $\pihat$ and there exist $C_1,D_1\in(0,\infty)$, such that
\[
\|\delta_x\Phat^n-\delta_y \Phat^n\|_V\leq C_1(\rho+D_1\epsilon)^{n}d_V(x,y).
\]
\end{theorem}

Theorem \ref{th:ergo} indicates that if $P$ is geometrically ergodic with ergodicity coefficient $\rho$ and $\Phat$ is $\epsilon$-close to $P$ as characterized by \eqref{eq:close2}, $\Phat$ is also geometrically ergodic. Moreover, the ergodicity coefficient of $\Phat$ is bounded above by $\rho+D_1\epsilon$.

In statistical applications, we are more interested in turning convergence results into error bounds for the Monte Carlo estimators. 
Central limit theorem of ergodic Markov processes were studied in 
 \cite{tierney1994markov, jones2004markov}, which provides asymptotic error quantifications.
 In practice, non-asymptotic bounds for finite values of $n$ may be more desirable. 
The following proposition is similar to Theorem 3 in \cite{joulin2010curvature}. We provide an explicit statement for the variance bound along with a simple proof for self-completeness. For simplicity, we assume the Markov process is initialized with the invariant measure, i.e., $\Xhat_0\sim\pihat$, so a burn-in period is not necessary. 

%

\begin{prop}
\label{pro:ergo}
Suppose $d_V(\delta_x\Phat^n,\delta_y \Phat^n)\leq \rhohat^n d_V(x,y)$ for some $\rhohat\in (0,1)$. Then, for any $f$ that is 1-Lipschitz under $d_V$,
\[
\left|\delta_x \Phat^n f -\pihat f\right|\leq  \rhohat^n(V(x)+\pihat V). 
\]
In addition, if we use $\fhat_M=\frac{1}{M}\sum_{k=1}^M f(\Xhat_k)$ as an estimator of $\pihat f$ starting from $\Xhat_0\sim \pihat$, 
\[
\E_{\pihat}\left[(\fhat_M-\pihat f)^2\right]\leq  \frac{2}{ (1-\rhohat)M}\E_{\pihat}\Big[|f(\Xhat_0)|(V(\Xhat_0)+\pihat V)\Big]. 
\]
\end{prop}

\section{The Spectral Gap Framework}
\label{sec:gap}
In this section, we discuss the spectral gap framework. 
We first introduce a few notations. For a transition kernel $P$ and density $\mu$, define
\[\|P\|_{\mu} = \max_{f:0<\|f\|_{\mu}<\infty} \frac{\|Pf\|_{\mu}}{\|f\|_{\mu}},\]
where $\|f\|_{\mu}^2=\langle f, f\rangle_{\mu}$. 
For two probability measures $\mu$ and $\nu$ on $\reals^d$, where $\nu$ is absolutely continuous with respect to $\mu$, define the $\chi^2$ divergence of $\nu$ from $\mu$ as:
\[D_{\chi^2}(\nu\|\mu)=\int \left(\frac{\nu(x)}{\mu(x)}-1\right)^2\mu(x)dx=\int 
\frac{\nu(x)^2}{\mu(x)}dx-1.
\]
For a transition kernel $P$ that is irreducible and reversible with respect to $\pi$, the spectral gap of $P$ is defined as \cite{hairer2014spectral}
\begin{equation}\label{eq:gap}
\kappa(P)=1-\sup\left\{\frac{\|Pf-\pi f\|_{\pi}^2}{\|f-\pi f\|_{\pi}^2}: f\in L^2(\pi), \var_\pi f \neq 0\right\}. 
\end{equation}
%
Note that by repeatedly applying \eqref{eq:gap}, we have for any $f\in L^2(\pi)$, 
\[\|P^n f -\pi f\|_{\pi}^2 \leq (1-\kappa(P))^n \|f-\pi f\|_{\pi}^2.\]
Thus, the larger the spectral gap, the faster $X_n$ converges to its invariant measure.

\begin{remark}
An alternative definition of the spectral gap takes the form
\[\kappa_a(P)=\inf\left\{\frac{\langle f, (I-P)f\rangle_{\pi}}{\var_{\pi} f}: f\in L^2(\pi), \var_\pi f \neq 0\right\}.\]
Note that the spectral gap defined in \eqref{eq:gap} can be viewed as the spectral gap of $P^2$ accordingly to this alternative definition, i.e.,
$\kappa(P)=\kappa_a(P^2)$.
\end{remark}

\subsection{General $\chi^2$ approximation and convergence} 



Our first result assumes that  $(X_n,P)$ has a spectral gap and  $\Phat$ is a close approximation of $P$:
\begin{thm}
\label{th:app_gap}
Suppose  $P$ is a reversible transition kernel with invariant measure $\pi$ and a spectral gap $\kappa(P)>0$ in the sense of \eqref{eq:gap}. 
$\Phat$ is another transition kernel satisfying $\|P-\Phat\|_\pi\leq \epsilon$ for a sufficiently small $\epsilon>0$.
Then, for any $a\in(0,1)$, there exists a constant $C$ such that 
the following holds with $\kappahat=(1-a)\kappa(P)-\frac{C\epsilon^2}{a}$:
\begin{enumerate}
\item For any $f\in L^2(\pi)$,  
\[\|\Phat^nf-\pi\Phat^nf\|_{\pi}^2 \leq (1-\kappahat)^{n}\var_\pi f \]
\item $\Phat$ has an invariant measure $\pihat$, which satisfies 
\[
|(\pihat-\pi \Phat^n)f|^2 \leq C\epsilon^2(1-\kappahat)^{n}\var_\pi f.
\]
Moreover, $D_{\chi^2}(\pihat\|\pi)\leq C\epsilon^2$. 
\end{enumerate}
\end{thm}

 Theorem \ref{th:app_gap} indicates that if $\Phat$ and $P$ are $\epsilon$-close to each other as quantified by 
$\|\Phat-P\|_{\pi}\leq \epsilon$, $\Phat$ has a stationary distribution $\pihat$. 
Moreover, $\pihat$ and $\pi$ are $\epsilon$-close to each other as quantified by 
$D_{\chi^2}(\pihat\|\pi)\leq C\epsilon^2$. 
We also note that showing that $\|\Phat^nf-\pi\Phat^nf\|_{\pi}^2 \leq (1-\kappahat)^{n} \var_\pi f$ is different from finding the spectral gap of $\Phat$,
since the latter would need a similar inequality but with $\pi$ replaced by $\pihat$. 
 In other words, Theorem \ref{th:app_gap} does not provide a spectral gap for $\Phat$. On the other hand, we can obtain error bound for Monte Carlo estimators using the bounds established in Theorem \ref{th:app_gap}:
\begin{proposition}
\label{prop:app_gap}
Under the same conditions as those in Theorem \ref{th:app_gap}, for any $f\in L^2(\pi)$ and any initial distribution  $\Xhat_0\sim \nu\ll\pi$, there exists a constant $C$ such that
\[
\left|\nu \Phat^n f  -\pihat f\right|^2\leq  (1-\kappahat)^{n}\var_\pi(f)\left(\sqrt{D_{\chi^2}(\nu\|\pi)+1}+C\epsilon\right)^2. 
\]
In addition, if $f$ is bounded, there exists a constant $C$ such that
\[
\E_{\pihat}[(\fhat_M-\pihat f)^2]\leq \frac{C}{M(1-(1-\kappahat)^{1/4})}\sqrt{\var_{\pihat}(f)\var_\pi (f)},
\]
where $\fhat_M=\frac{1}{M}\sum_{k=1}^M f(\Xhat_k)$.
\end{proposition}

\subsection{Spectral gap with density ratio bounds}
In this section, we show that stronger results can be established if we can bound the ratio between the invariant densities $\pi$ and $\pihat$.
Such a bound is assessable if we have an explicit characterization of $\pihat$.
For example, in Bayesian inverse problems, $\pi(x)\propto p_0(x)\exp(-\frac12\|G(x)-y\|^2)$ while 
$\pihat(x)\propto p_0(x)\exp(-\frac12\|\widehat{G}(x)-y\|^2)$. In this case, a density ratio bound can be obtained if $\|G(x)-\widehat{G}(x)\|$ is bounded, which is practically feasible by using an accurate numerical approximation of $G$. 

\begin{thm} \label{th:gap_ratio}
Suppose $P$ and $\Phat$ are two reversible transition kernels with invariant densities $\pi$ and $\pihat$ respectively. 
We further assume $\pi(x)/\pihat(x)\in [(1+\epsilon)^{-1}, 1+\epsilon]$ and $\|P-\Phat\|_{\pi}\leq\epsilon$.
Then, there exists a universal constant $C$ such that
\[
\kappa(\Phat)\geq \kappa(P)-C\epsilon. 
\]
\end{thm}
Based on the spectral gap, we have the following non-asymptotic Monte Carlo error bound.
\begin{proposition}
\label{pro:gapproof}
Suppose $(\Xhat_n,\Phat)$ has a spectral gap $\kappahat$. Suppose the initial distribution is $\nu$, i.e., $\Xhat_0\sim \nu$. Then, 
\[
\left|\E f(\Xhat_n) -\pihat f\right|^2\leq  (1-\kappahat)^{n}\var_{\pihat}f(D_{\chi^2}(\nu\|\pihat)+1).
\]
In addition, if $\nu=\pihat$, 
\[
\E_{\pihat}[(\fhat_M-\pihat f)^2]\leq \frac{C}{M(1-(1-\kappahat)^{1/4})}\sqrt{\var_{\pihat}f ~ \var_\pi f},
\]
where  $\fhat_M=\frac{1}{M}\sum_{k=1}^M f(\Xhat_k)$.
\end{proposition}

Before we conclude our discussion of the spectral gap framework, we remark that even though the condition $\|P-\Phat\|_\pi\leq \epsilon$ is reasonable  for the spectral gap analysis, it can be hard to verify directly in some applications. 
To remedy this issue, 
 the next proposition shows that we can bound $\|P-\Phat\|_\pi$ through a bound for 
$\|\delta_x P - \delta_x\Phat\|_{TV}$, which can be easier to obtain using coupling tools. 

\begin{prop} \label{prop:pert}
Suppose there exists a $\pi$-measurable function $V:\Omega\rightarrow[1,\infty)$ such that
$\|\delta_x P - \delta_x\Phat\|_{TV}\leq \epsilon V(x)$.
In addition, suppose $\frac1a\leq \pi(x)/\hat\pi(x)\leq a$ for some constant $a>0$.
Then, 
\[
\|P-\Phat\|_{\pi}\leq \sqrt{2(1+a^2)}\sqrt{\epsilon}\|V\|_{\pi}^{1/2}.
\]
\end{prop}

\section{Application: Metropolis Hasting MCMC on perturbed densities} \label{sec:Metro}
Random walk Metropolis (RWM) and Metropolis adjusted Langevin algorithm (MALA) are two popular MCMC samplers when it comes to sampling a generic density $\pi$. 
Many existing works have already studied their spectral gap under suitable conditions on $\pi$ \cite{roberts1996exponential,hairer2014spectral,dwivedi2018log}. 
When implementing these samplers, it is often the case that we only have access to an approximation of $\pi$, which we denote as $\pihat$.
In this section, we will demonstrate how to apply our analysis framework to establish proper bounds for the spectral gap of the ``numerical" RWM and MALA. 

In fact, we can develop some general results for Metropolis Hasting (MH) type of Monte Carlo algorithm. Assume the proposals are given by some smooth transition density $R(x,x')$. Due to the possibility of rejection, MH Monte Carlo transition densities can be written as  $P(x,x')=\alpha(x)\delta_x(x')+\beta(x,x')$ with
\begin{equation}
\label{eq:MH}
 \beta(x,x')=\min\left\{\frac{\pi(x')R(x',x)}{\pi(x)}, R(x,x')\right\},\quad \alpha(x)=1-\int \beta(x,x')dx'.
\end{equation}
The perturbed transition density can be written as $\Phat(x,x')=\alphahat(x)\delta_x(x')+\betahat(x,x')$.
We provide some sufficient conditions under which the difference between $P$ and $\Phat$ is of order $\epsilon$.
\begin{lem}
\label{lem:MHtransition}
If the transition density is of the form  $P(x,x')=\alpha(x)\delta_x(x')+\beta(x,x') $ with $\nu(x) P(x,x')=\nu(x') P(x',x)$, suppose  
$\Phat(x,x')=\hat{\alpha}(x)\delta_x(x')+\hat{\beta}(x,x')$ with 
\[
|\hat{\alpha}(x)-\alpha(x)|\leq C\epsilon ~\mbox{ and }~ (1-C\epsilon )\beta(x,x')\leq \hat{\beta}(x,x')\leq (1+C\epsilon) \beta(x,x').
\]
for some constant $C\in(0,\infty)$. Then, there exists a constant $C_1\in(0,\infty)$ such that  $\|P-\Phat\|_\nu\leq C_1\epsilon.$
\end{lem}

%
%
%
%
\subsection{Random walk Metropolis}
RWM considers implementing the MH procedure on random walk proposals. That is, we use 
\[
R(x,x')=\frac{1}{\sqrt{2\pi h }^{d}}\exp\left(-\frac1{4h}\|x'-x\|^2\right)
\]
in \eqref{eq:MH}. It is worth noting that using a perturbed density $\pihat$ does not affect this proposal. 
\begin{prop}
\label{prop:l2MWN}
For RWM, if $\sup_x|\log\pi(x)-\log\hat\pi (x)|\leq C\epsilon$, there is a constant $C_1$ so that
\[
\|P_{RWM} - \Phat_{RWM}\|_{\pi}\leq C_1\epsilon.
\]
\end{prop}

If the original RWM has a spectral gap and $\sup_x|\log\pi(x)-\log\hat\pi (x)|\leq C\epsilon$, 
then Proposition \ref{prop:l2MWN} together with Theorem \ref{th:gap_ratio} implies that the perturbed
RWM has a proper spectral gap as well.

\subsection{Metropolis adjusted Langevin algorithm}
MALA considers implementing the MH procedure on  proposals following the Langevin diffusion. That is, we use 
\[
R(x,x')=\frac{1}{\sqrt{4\pi h }^{d}}\exp\left(-\frac1{4h}\|x'-x-h\nabla \log \pi(x)\|^2\right)
\]
in \eqref{eq:MH}. Using a perturbed density $\pihat$ does change this proposal. We discuss the perturbation in two separate cases. 
In particular, we shall verify that the condition $\|P_{MALA}-\Phat_{MALA}\|_{\pi}\leq \epsilon$ holds under appropriate assumptions on $\pihat$ in the two cases.
Then, if $P_{MALA}$ has a spectral gap, the numerical sampler $\Phat_{MALA}$ has a proper spectral gap as well.

\subsubsection{Bounded domain}
When the support of $\pi$ and $\pihat$ are bounded, the analysis is quite straight forward with Lemma \ref{lem:MHtransition}.
\begin{prop}
\label{prop:l2MALA}
For MALA, if $\sup_x|\log\pi(x)-\log\hat\pi(x)|\leq C\epsilon$, $\sup_x\|\nabla\log\pi(x)-\nabla\log\hat\pi(x)\|\leq C\epsilon$, and the support of $\pi$ and $\pihat$ are bounded, then 
\[
\|P_{MALA}-\Phat_{MALA}\|_{\pi}=O(\epsilon).
\]
\end{prop}

\subsubsection{Unbounded support}
When the support of the density is unbounded, directly bounding $\|P_{MALA}-\Phat_{MALA}\|_{\pi}$ becomes difficult. Instead, we 
consider establishing $\| \delta_xP- \delta_x\Phat\|_{TV}=O(\epsilon)$. 
\begin{prop}\label{prop:MALA_TV}
For MALA, if $\log \pi$ is Lipschitz, $\sup_x|\log\pi(x)-\log\hat\pi(x)|\leq L_\pi\epsilon$, and moreover $\sup_x\|\nabla\log\pi(x)-\nabla\log\pihat(x)\|\leq L_\pi\epsilon$, 
for any $\delta>0$, there exists $C_{\delta}\in(0,\infty)$, such that for $h<(\frac{5L_{\pi}}{\delta} + 20L_{\pi})^{-1}$,
\[\| \delta_xP- \delta_x\Phat\|_{TV}\leq C_{\delta}\epsilon\exp(\delta \|x\|^2).\]
\end{prop}
When $\pi(x)$ is sub-Gaussian, we can find a $\delta>0$ such that $V(x)=\exp(\delta \|x\|^2)$ is $L_2$-integrable under $\pi$. 
Then Proposition \ref{prop:pert} indicates that $\|P- \Phat\|_{\pi}=O(\sqrt{\epsilon})$.

\section{Application: Parallel Tempering with Perturbed Densities} \label{sec:RE}

In this section, we demonstrate how to apply our framework to parallel tempering (PT) algorithms \cite{earl2005,tawn2019accelerating,tawn2020weight}. These algorithms are also referred to as the replica exchange methods \cite{sugita1999replica,dupuis2012infinite,dong2021replica}. Compare with regular MCMC samplers like RWM and MALA, PT tries to sample a multiple tempered version of the target density. Such design can improve the convergence rate on densities with multiple isolated modes. 

To implement PT,  a sequence of distributions $\pi_0,\ldots, \pi_K$ are considered where  the last one is the target density $\pi_K =\pi$. The first density $\pi_0$ is usually a distribution that is easy to draw samples from. The intermediate distributions, $\pi_k$'s $1\leq k\leq K-1$, are set up so that the two neighboring densities are similar to each other.  A common choice for the intermediate distributions is to consider interpolations between $\pi_K$ and $\pi_0$:
\[
\pi_k(x)\propto\pi^{\beta_k}(x)\pi_0^{1-\beta_k}(x),
\]
where $0=\beta_0<\beta_1<\ldots <\beta_K=1$ is a sequence of parameters. PT intends to generate samples from the product density 
\[
\Pi=\pi_0\times \pi_1\times\cdots\times \pi_K \mbox{ on $\reals^{d(K+1)}$.}
\] 
To do so, its iterations consist of $K+1$ parts, i.e.,  $X_n=(X^0_{n},\ldots,X^K_{n})$, and the updating rule is given by  the following two steps.
\begin{enumerate}
\item Updating each $X^k_{n}$ to $X^{k}_{n+1}$ according to a transition kernel $M_k$, whose stationary distribution is $\pi_k$.
In practice, $M_k$ is often taken as the transition kernel obtained by repeating RWM or MALA update for $t_k$ steps. That is $M_k=P_{RWM}^{t_k}$ or  $M_k=P_{MALA}^{t_k}$. 
\item Pick an index $k\in \{0,\ldots,K-1\}$ uniformly at random and swap the values of $X^k_{n+1}$ and $X^{k+1}_{n+1}$ with probability $\alpha_k(X_{n+1}^k,X^{k+1}_{n+1})$, where
\[
\alpha_k(x,x')=\min\left\{1, \frac{\pi_k(x')\pi_{k+1}(x)}{\pi_{k}(x)\pi_{k+1}(x')}\right\}.\]
\end{enumerate}
The pseudo code of PT is given in Algorithm \ref{alg:re}.   
%
%
%
\begin{algorithm}[ht]
\caption{Parallel Tempering} \label{alg:re}
\begin{algorithmic}
 \State{Input: Replica counts $K$, target densities $\pi_k$ for $k=0,\dots, K$, transition kernels $M_k$ targeting $\pi_k$.}
 \State{Output: $(x_t^k)_{t=0,\ldots,T_k, k=0,\ldots,K} $ as samples from $\pi_k$}
\State{Initialize $x_0^k$ for all $k$}\;
 \For{$ t=0$ to $T$}
 \For{$ k=0$ to $K$} \%Run MCMC at each level\;
\State{Generate $x^k_{t+1}\sim M_k(x_t^k,\,\cdot\,)$.}
\EndFor

 \%Consider swapping at a random level\;
 \State{Let $k$ be a random index in $\{0,\ldots,K-1\}$\;}
\State{Let $U$ be a random sample from Unif$[0,1]$\;}
 \If {$U<a_k(x^k_{t+1},x^{k+1}_{t+1})$}
 \State{$(x^k_{t+1},x^{k+1}_{t+1})=(x^{k+1}_{t+1},x^{k}_{t+1})$}
\EndIf
\EndFor
\end{algorithmic}
\end{algorithm}
The exchange procedure can be described by the transition probability on $R^{(K+1)d}\times R^{(K+1)d}$:
\[
Q_k(x,x)=1-\alpha_k(x^k, x^{k+1}),\quad Q_k(x,S_k(x))=\alpha_k(x^k, x^{k+1}),
\]
where $S_k(x)=(x^0,\ldots, x^{k-1},x^{k+1},x^k, x^{k+2},\ldots, x^K)$. With a little abuse of notation, we write the transition kernel as 
$Q_k$ as well, i.e., $Q_kf(x)=Q_k(x,x)f(x)+Q_k(x,S_k(x))f(S_k(x))$. 
The transition kernel of PT can then be written as 
\begin{equation}
\label{eq:PTker}
P=\left(M_{0}\otimes \cdots  \otimes M_K\right) \left(\frac{1}{K}\sum_{\mathbf{k}\in \{0,\ldots,K-1\}} Q_{k}\right),
\end{equation}
where the direct product of two transition kernels is given by 
\[
M_0\otimes M_1f(x^0,x^1)=\int\int M_0(x^0,y^0)M_1(x^1,y^1) f(y^0,y^1) dy^0dy^1. 
\]
The spectral gap of $P$ in \eqref{eq:PTker} has been studied in \cite{woodard2009conditions}. Assume the state space can be partition into $\reals^d=\cup_{i=1}^J A_j$, it is shown that $\kappa(P)$ can be seen as the product of three elements: 1) the maximal spectral gap when sampling $\pi_k$, $k\geq 1$, constrained on one piece $A_j$;  2) the spectral gap when sampling $\pi_0$ using $M_0$; and 3) the density ratio: $\pi_k(A_j)/\pi_{k+1}(A_j)$. In particular, if $\pi_0$ is easy to sample, $\pi_k$ is not so different from $\pi_{k+1}$,  and the sampling of $\pi_k$ constrained on $A_j$ is efficient, then PT can be highly efficient. 

When implementing PT numerically, we may not have access to the exact values of $\pi_k$, but only an $\epsilon$-approximation, which we denote as $\pihat_k$. Then, the corresponding PT uses a sampler $\Mhat_k$ with invariant measure $\pihat_k$ at each replica, while the exchange probability is given by 
\[
\alphahat_k(x,x')=\min\left\{1, \frac{\pihat_k(x')\pihat_{k+1}(x)}{\pihat_{k}(x)\pihat_{k+1}(x')}\right\}.
\]
The corresponding transition kernel can be written as
\[
\Phat=\left(\Mhat_{0}\otimes \cdots  \otimes \Mhat_K\right) \left(\frac{1}{K}\sum_{\mathbf{k}\in \{0,\ldots,K-1\}} \Qhat_{k}\right).
\]
It is natural to ask whether this numerical PT will inherit the spectral gap of $P$. The next result together with Theorem \ref{th:gap_ratio} indicates that under appropriate regularity conditions on $\pihat_k$'s, the numerical PT also has a proper spectral gap. 
\begin{prop} \label{prop:re}
Suppose for each replica the target distribution satisfies $\sup_x|\log\pihat(x)-\log\pi(x)|\leq \epsilon$ and the transition kernel satisfies $\|P_k-\Phat_k\|_{\pi_k}\leq \epsilon$, then the transition kernel of PT satisfies the following for some constant $C$:
\[
\|P-\Phat\|_{\Pi}\leq C\epsilon.
\]
\end{prop}

Before we prove Proposition \ref{prop:re}, we first prove two auxiliary lemmas.
The first lemma shows that different compositions of approximated transition kernels yield approximation kernels of similar accuracy. 
In particular, it helps us establish the condition $\|P_k-\Phat_k\|_{\pi_k}\leq \epsilon$ in Proposition \ref{prop:re}  if we use $M_k=P_{RWM}^{t_k}$ or  $M_k=P_{MALA}^{t_k}$.
\begin{lem}
\label{lem:l2comp}
\begin{enumerate}[1)]
\item For two transition kernels $R$ and $S$, both with invariant measure $\nu$,
if $\|R-\Rhat\|_\nu\leq C\epsilon$ and $\|S-\Shat\|_\nu\leq C\epsilon$, then there is a constant $C'$ so that
\[
\|RS-\Rhat\Shat\|_{\nu}\leq C'\epsilon.
\]
\item For two transition kernels $R_1$ and $R_2$ with invariant measure $\nu_1$ and $\nu_2$ respectively,
if $\|R_1-\hat R_1\|_{\nu_1}\leq C\epsilon$ and $\| R_2-\hat R_2\|_{\nu_2}\leq C\epsilon$,  then there is a constant $C'$ so that
\[
\|R_1\otimes R_2 -\Rhat_1\otimes \Rhat_2\|_\nu \leq C'\epsilon,
\]
where $\nu=\nu_1\times\nu_2$ is the joint invariance distribution.
\item For $n$ transition kernels $S_1, S_2, \dots, S_n$, all with invariant measure $\nu$,
if $\|S_i-\Shat_i\|_\nu\leq C\epsilon$ for $i=1,\dots, n$, then for $U=\frac1n\sum_{i=1}^n S_i$ and $\Uhat=\frac1n\sum_{i=1}^{s} \Shat_i$, 
 there is a constant $C'$ so that
\[
\|U - \Uhat\|_\nu\leq C'\epsilon.
\]
\end{enumerate}
\end{lem}

The second lemma establishes proper bounds for the swapping transition.
\begin{lem}
\label{lem:REl2}
Let $Q$ be a transition probability of form:
\[ 
Q(x,S(x))=a(x,S(x)),\quad Q(x,x)=1-a(x,S(x)),\]
where $S(x)$ is some given map. Suppose $Q$ is reversible with a density $\nu$, i.e. 
\[
\nu(x) Q(x,S(x))=\nu(S(x)) Q(S(x),x).
\] 
Similarly, let $\Qhat$ denote the transition probability of form 
\[ 
\Qhat(x,S(x))=\widehat{a}(x,S(x)),\quad \Qhat(x,x)=1-\widehat{a}(x,S(x)),
\]
reversible with $\nuhat$. If for some constant $C$, $(1-C\epsilon)a(x,S(x))\leq \widehat{a}(x,S(x))\leq (1+C\epsilon)a(x,S(x))$, then  
\[
\|(Q-\Qhat) f\|_{\nu}\leq 2C \epsilon  \|f\|_\nu.
\]
\end{lem}

\section{Numerical examples}
\label{sec:num}
In this section, we present some numerical examples based on the predator-prey system to illustrate the theoretical results developed in the preceding sections. 

\subsection{Predator-prey system}
We consider inferring the parameters of a system of ordinary differential equations (ODEs) that models the predator-prey system \cite{lotka1925elements}. 
Denoting the populations of prey and predator by $( \gamma_p, \gamma_q )$, 
the populations change over time according to the pair of coupled ODEs:
\begin{align}
\frac{d \gamma_p}{dt} & = r \gamma_p \left(1 - \frac{\gamma_p}{K}\right) - s \left(\frac{\gamma_p \,\gamma_q }{w + \gamma_p}\right), \nonumber \\
\frac{d \gamma_q}{dt} & = u \left(\frac{\gamma_p \,\gamma_q }{w + \gamma_p}\right) - v \gamma_q, \label{eq:pp_eq}
\end{align}
with initial conditions $\gamma_p(0)$ and $\gamma_q(0)$.
$r$, $K$, $a$, $s$, $u$, and $v$ are model parameters that control the dynamics of the populations of prey and predator. In the absence of the predator, the population of prey evolves according to the logistic equation, which is characterized by $r$ and $K$. In the absence of the prey, the population of predator has an exponential decay rate $v$. The additional parameters $s, w$, and $u$ characterize the interaction between the predator population and the prey population.

In the inference problem, we want to estimate both the model parameters and the initial conditions. In this case, we have $d = 8$ and denote 
\[
\theta = ( \gamma_p(0), \gamma_q(0), r, K, a, s, u, v ).
\]
A commonly used prior for this problem is a uniform distribution over a hypercube $(a_1,b_1) \times \cdots \times (a_d,b_d) $ (see, e.g., \cite{parno2018transport}). Here, we set $a_i = 10^{-3}$ and $b_i = 2\times 10^2$ for all $i$. 
Noisy observations of both $\gamma_p(t; \theta)$ and $\gamma_q(t; \theta)$ at times regularly spaced at $m = 20$ time points in $t \in [2, 40]$ are used to infer $\theta$. This defines a so-called forward model
\[
F(\theta) = [\gamma_p(t_1; \theta), \gamma_q(t_1; \theta), \ldots, \gamma_p(t_m; \theta), \gamma_q(t_m; \theta)],
\]
that maps a given parameter $\theta$ to the observables. 
The observables are perturbed with independent Gaussian observational errors with mean zero and variance $4$. 
A ``true'' parameter
\[
\theta_{\rm true} = [50, 5, 0.6, 100, 1.2, 25, 0.5, 0.3]^\top
\]
is used to generate the synthetic observed data set, which is denoted by $y$. The trajectories of $\gamma_p(t; \theta_{\rm true})$ and $\gamma_q(t; \theta_{\rm true})$ together with the synthetic data set are shown in Figure \ref{fig:pp_setup}.

\begin{figure}
\centerline{\includegraphics[width=\textwidth]{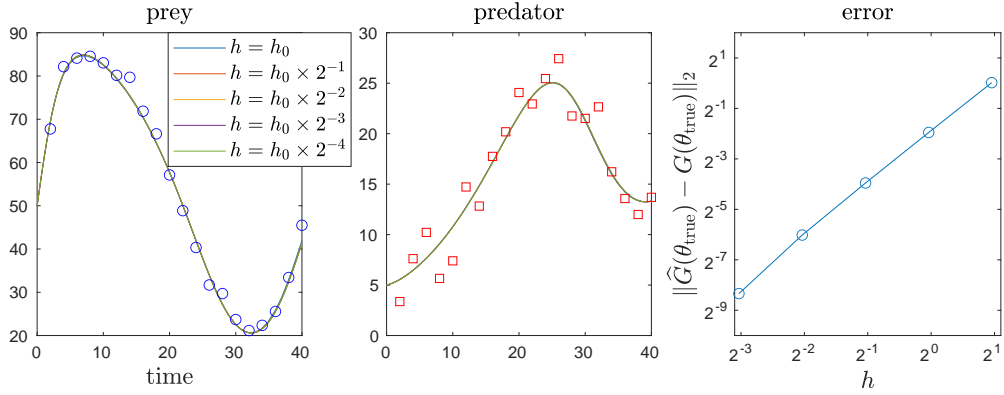}}
\caption{Left and middle: the trajectories of $\gamma_p(t; \theta_{\rm true})$ and $\gamma_q(t; \theta_{\rm true})$ computed using the second order Runge--Kutta method with different time step size $h$. Right: the $L_2$ error of the model outputs with different time step size $h$. Here $G(\theta_{\rm true})$ is computed using $h = h_0 \times 2^{-6}$. The trajectories computed by the time step size $h = h_0 \times 2^{-6}$ is used to generate synthetic data set. The observed data sets of the prey and predator are shown as circles and squares, respectively. }\label{fig:pp_setup}
\end{figure}

\begin{figure}
\centerline{\includegraphics[width=\textwidth]{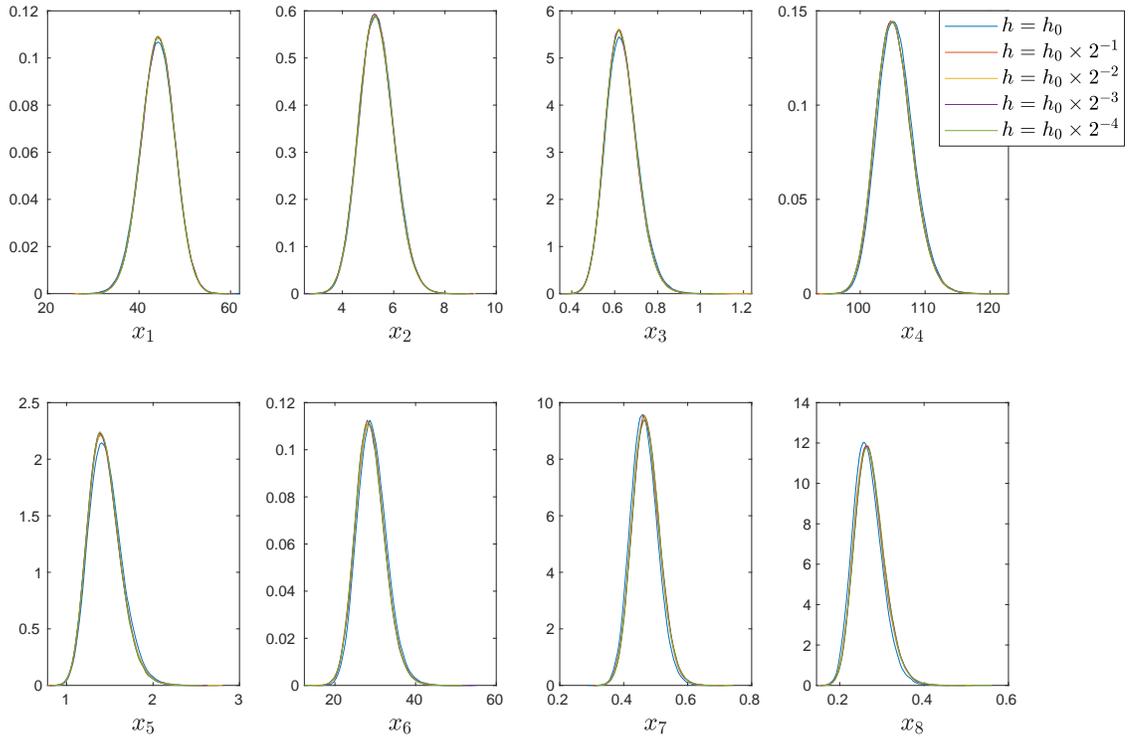}}
\caption{Marginal distributions of perturbed posteriors defined by various time step sizes.}\label{fig:marginals}
\end{figure}

To avoid rejections caused by proposal samples that fall outside of the hypercube, we further consider the prior distribution as the pushforward of the standard Gaussian measure with the probability density function 
\[
p_0(x) = (2\pi)^{-d/2} \exp\left( -\frac12 \|x\|^2  \right)
\]
under a diffeomorphic transformation $T : \R^d \rightarrow \R^d$ that 
maps each coordinate
\[
\theta_i = T_i(x_i) = a_i + \frac{b_i - a_i}{\sqrt{2\pi} }\int_{-\infty}^{x_i} \exp\left( -\frac12 z^2  \right) d z.
\]
In other words, $p_0(x)$ is the prior distribution for the transformed parameter $x = T^{-1}(\theta)$. Writing $G(x) = F(T(x))$, our goal is to characterize the posterior distribution
\[
\pi(x) \propto p_0(x) \exp\left( - \frac{1}{8} \|G(x) - y\|^2 \right).
\]
The system of ODEs in \eqref{eq:pp_eq}, and hence the function $G(x)$, has to be numerical solved by some ODE solvers. Here we use the second order explicit Runge--Kutta method with time step size $h$ to solve \eqref{eq:pp_eq}. As shown in Figure \ref{fig:pp_setup}, the trajectories of $\gamma_p(t; \theta_{\rm true})$ and $\gamma_q(t; \theta_{\rm true})$ converge as $h \rightarrow 0$. The numerical solver, which is characterized by the step size $h$, defines the approximate model $\widehat{G}(x)$ and the approximate posterior density $\widehat{\pi}(x)$. Figure \ref{fig:marginals} shows the estimated marginal distributions (using Algorithm \ref{alg:re}) of perturbed posteriors defined by various time step sizes. Here we observe that as $h$ decreases, the estimated marginal distributions almost overlap each other, which suggests that the perturbed distributions converge as the discretized model $\widehat{G}$ converges. 

\subsection{MCMC results}

To validate the theoretical results on Metropolis--Hasting MCMC on perturbed densities in Section \ref{sec:Metro}, we first simulate the RWM algorithm with invariant densities  $\widehat{\pi}(x)$ defined by various time step sizes as shown in Figure \ref{fig:pp_setup}. All the Markov chains in this set of simulation experiments are generated using the same Gaussian random walk proposal distribution. The resulting autocorrelation times are shown in Figure \ref{fig:pp_rw}.
Then, we simulate MALA with invariant densities  $\widehat{\pi}(x)$ defined by the same set of time step sizes. The resulting autocorrelation times are shown in Figure \ref{fig:pp_mala}. Again, all the Markov chains are generated using the same proposal distribution. 
For both algorithms, we simulate each Markov chain for $10^{6}$ iterations after discarding burn-in samples. Each Markov chain simulation is repeated for $20$ times with different initial states. The results in Figures \ref{fig:pp_rw} and \ref{fig:pp_mala} summarizes the mean and the $\pm 2$ standard deviations of the autocorrelation times. As established in our theoretical analysis, for both algorithms, the resulting Markov chains targeting on various approximate posterior densities produce similar autocorrelation times. This provides empirical evidence that the spectral gaps of the approximate transition kernels defined by MRW or MALA  converge as the discretization step size $h \rightarrow 0$.

To validate the theoretical results on the parallel tempering with perturbed densities in Section \ref{sec:RE}, we simulate Algorithm \ref{alg:re} with the same Gaussian random walk as in RWM. For each of the invariant densities  $\widehat{\pi}(x)$ defined by various time step sizes, we set $K=4$, and the intermediate distributions take the form
\[
\widehat{\pi}_k(x) \propto p_0(x)  \exp\left( - \frac{\beta_k}{8} \|\widehat{G}(x) - y\|^2 \right),
\]
where $\beta_k = 1 + \alpha^{-K} - \alpha^{-k}$ with $\alpha = 1.3$ and $k = \{0,1,2,3,4\}$. Here $\beta_k$ is an increasing sequence such that $\beta_K = 1$. The same Gaussian random walk is used across all replicas to simulate the Markov chain. The autocorrelation times of the resulting Markov chains are shown in Figure \ref{fig:pp_rpex}. Similar to the previous experiments, the resulting Markov chains targeting on various approximate posterior densities produce similar autocorrelation times. This provides empirical evidence that the spectral gaps of the approximated transition kernel induced by Algorithm \ref{alg:re} converge as the discretization step size $h \rightarrow 0$.

\begin{figure}
\centerline{\includegraphics[width=\textwidth]{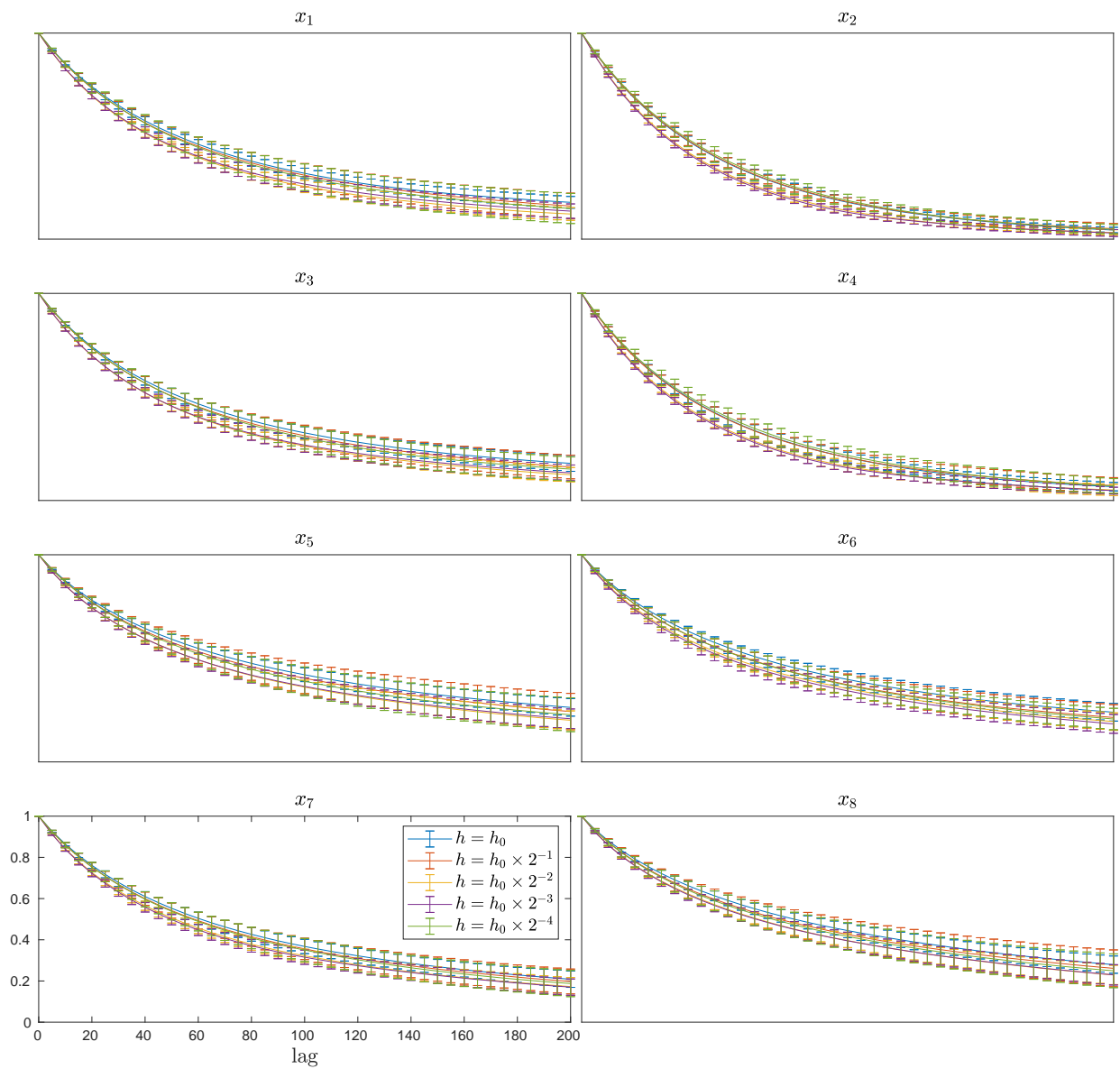}}
\caption{Autocorrelation time of each of the parameter Markov chains simulated by the RWM algorithm. Here different colored lines represent Markov chains targeting invariant measures defined by different time discretization steps.}\label{fig:pp_rw}
\end{figure}

\begin{figure}
\centerline{\includegraphics[width=\textwidth]{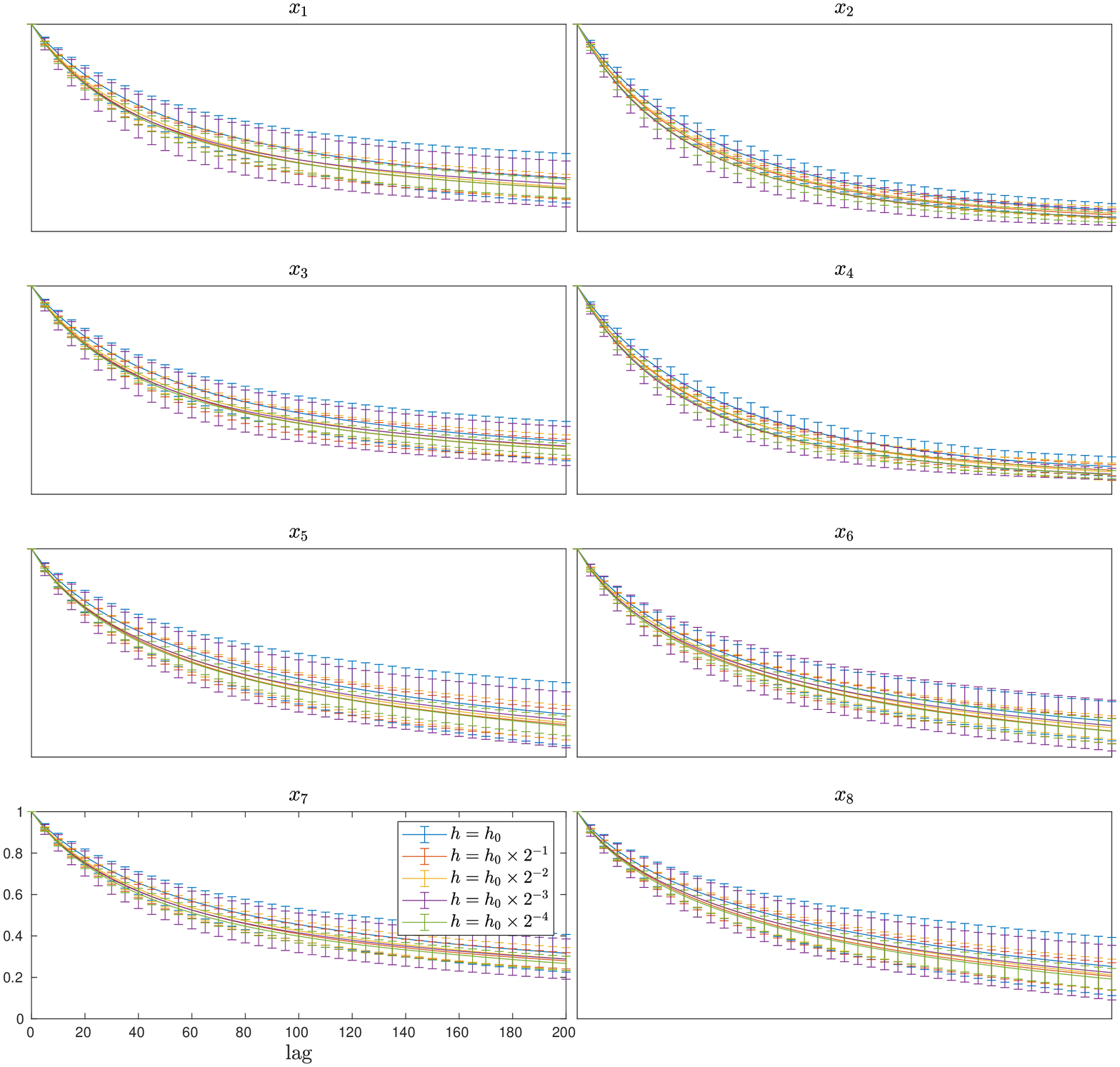}}
\caption{Autocorrelation time of each of the parameter Markov chains simulated by the MALA algorithm. Here different colored lines represent Markov chains targeting invariant measures defined by different time discretization steps.}\label{fig:pp_mala}
\end{figure}

\begin{figure}
\centerline{\includegraphics[width=\textwidth]{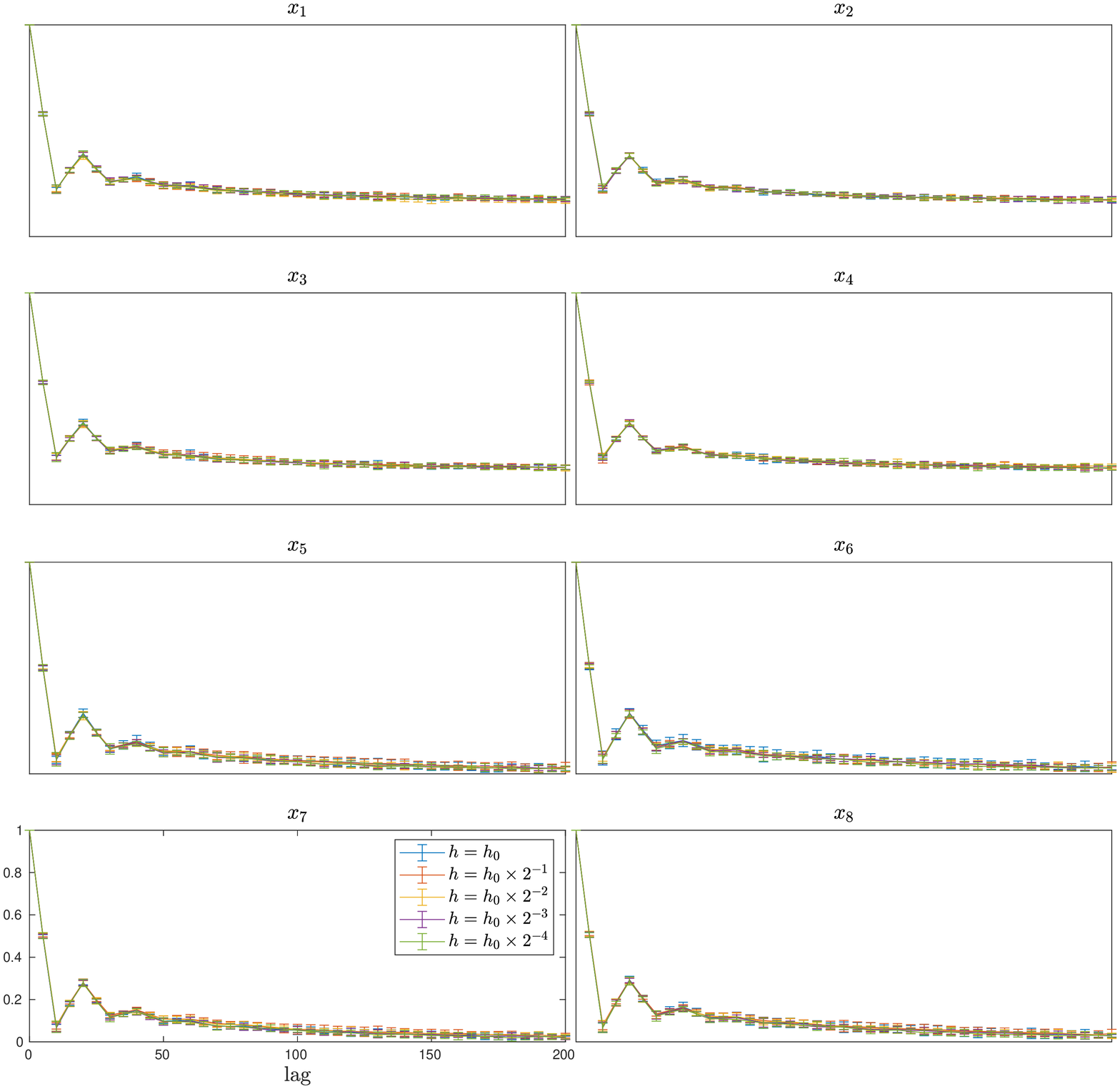}}
\caption{Autocorrelation time of each of the parameter Markov chains simulated by Algorithm \ref{alg:re}. Here different colored lines represent Markov chains targeting invariant measures defined by different time discretization steps.}\label{fig:pp_rpex}
\end{figure}

\section{Conclusion}
In this paper, we quantify the convergence speed and the approximation accuracy of numerical MCMC samplers under two general frameworks: ergodicity and spectral gap. Our results can be easily applied to study the efficiency and accuracy of various sampling algorithms. In particular, we demonstrate how to apply our framework to study Metropolis Hasting MCMC algorithms and parallel tempering Monte Carlo algorithms. These results are validated by numerical simulations on a Bayesian inverse problem based on the predator-prey model. 

\appendix

\section{Proof for the ergodicity framework}
\begin{proof}[Proof of Theorem \ref{th:ergo}]
Let $Q_x$ be the optimal coupled measure between $\delta_xP$ and $\delta_x\Phat$. Then,
\begin{equation*}\begin{split}
|PV(x)-\Phat V(x)|&=\left|\int Q_{x}(dx',dy') (V(x')-V(y'))\right|\\
&\leq \left|\int Q_{x}(dx',dy') (V(x')+V(y'))1_{x'\neq y'}\right|\leq \|\delta_x P- \delta_x \Phat\|_V.
\end{split}\end{equation*}
Next, as $\|\delta_x P-\delta_x \Phat\|_V\leq \epsilon V (x)$, we have
\[
|PV(x)-\Phat V(x)|\leq \epsilon V(x).
\]
In addition, because $V$ is a Lyapunov function under $P$, 
\[
\Phat V(x)\leq PV(x)+\epsilon V(x)\leq (\lambda+\epsilon) V(x)+L.
\]
As $(\lambda+\epsilon)\in(0,1)$ for $\epsilon$ small enough,
$V$ is a Lyapunov function under $\Phat$ with parameters $\lambda+\epsilon$ and $L$.\\


We next establish a bound for $\|\delta_x\Phat^n - \delta_y \Phat^n\|_V$, $x\neq y$ using
\begin{equation}
\label{tmp:2}
\|\delta_x\Phat^n-\delta_y \Phat^n\|_V\leq \|\delta_x\Phat^n-\delta_x P^n\|_V+\|\delta_xP^n-\delta_y P^n\|_V+\|\delta_yP^n-\delta_y \Phat^n\|_V.
\end{equation}
For $\|\delta_x\Phat^n-\delta_x P^n\|_V$, by Theorem \ref{thm:rs}, we have
\[
\|\delta_x\Phat^n-\delta_x P^n\|_V\leq \frac{C\epsilon }{1-\rho}\left(V(x)+\frac{L}{1-\lambda-\epsilon}\right) \leq C\epsilon V(x)
\]
for some $C$, because $V(x)\geq 1$.
A similar bound holds for $\|\delta_yP^n-\delta_y \Phat^n\|_V$ as well, i.e.,
\[
\|\delta_yP^n-\delta_y \Phat^n\|_V\leq \frac{C\epsilon }{1-\rho}\left(V(y)+\frac{L}{1-\lambda-\epsilon}\right) \leq C\epsilon V(y). 
\]
Then, for any $l\leq N$, if we let $D_1=\frac{C}{N\rho}$, \eqref{tmp:2} leads to 
\begin{equation}\label{eq:ergo_bd1}
\begin{split}
\|\delta_x\Phat^{l+N} - \delta_y \Phat^{l+N}\|_V &\leq\rho^{l+N} d_V(x,y)+C\epsilon (V(x)+V(y))\\
&= \left(\rho^{l+N} + C\epsilon\right)d_V(x,y)\leq (\rho+D_1\epsilon)^{l+N} d_V(x,y).
\end{split}\end{equation}
Next, let $\Qhat^{k}_{x,y}$ be the optimal coupled measure between $\delta_x \Phat^{kN}$ and $\delta_y \Phat^{kN}$. 
Then,
\begin{equation} \label{eq:ergo_bd2}
\begin{split}
\|\delta_x \Phat^{kN}-\delta_y \Phat^{kN}\|_V 
&\leq \int \Qhat_{x,y}^{k-1}(dx',dy') \|\delta_{x'} \hat P^N - \delta_{y'} \hat P^N\|_V\\
&\leq \int \Qhat_{x,y}^{k-1}(dx',dy') (\rho+D_1\epsilon)^N d_V(x',y')\\
&\leq (\rho+D_1\epsilon)^{kN} d_V(x,y).
\end{split}\end{equation}
For any $n\geq N$, we can write $n=kN+N+l$, for $k,l\in \mathbb{Z}_0^+$, and 
\begin{equation*}\begin{split}
\|\delta_x \Phat^{n}-\delta_y \Phat^{n}\|_V
&\leq \int Q^k_{x,y}(dx',dy')\|\delta_{x'} \hat P^{N+l} - \delta_{y'} \hat P^{N+l}\|_V\\
&\leq C_1(\rho+D_1\epsilon)^{N+l} \int \Qhat^k_{x,y}(dx',dy')d_V(x',y') \mbox{ by \eqref{eq:ergo_bd1}}\\
&\leq C_1(\rho+D_1\epsilon)^{l+kN+N}d_V(x,y) \mbox{ by \eqref{eq:ergo_bd2}}\\
&=C_1(\rho+D_1\epsilon)^{n}d_V(x,y).
\end{split}\end{equation*}

Lastly, we show that $\Phat$ has a unique  invariant measure $\pihat$. Fix a point $x$, consider a sequence  $\{\delta_x\Phat^n, n=1,2,\ldots\}$. Note that 
\begin{align*}
\|\delta_x\Phat^n-\delta_x\Phat^{n+1}\|_V&\leq \int \|\delta_x\Phat^n-\delta_y\Phat^{n}\|_V\Phat(x,dy)\\\
&\leq C_1(\rho+D_1\epsilon)^{n}\E[d_V(x,\Xhat_1)]\leq C_1 (\rho+D_1\epsilon)^n[ (\lambda+1+\epsilon)V(x)+L].
\end{align*}
This implies that $\delta_x\Phat^n$ is a Cauchy sequence in $d_V$ and the total variation distance. Therefore, the sequence has a limit, which we denote by $\pi_x$. Next, we show that $\pi_x=\pi_y$: 
\[
\|\pi_x-\pi_y\|_{V}\leq \|\pi_x-\delta_x\Phat^n\|_V+\|\delta_y\Phat^n-\pi_y\|_{V}+\|\delta_x\Phat^n-\delta_y\Phat^n\|_V\to 0 \mbox{ as $n\to 0$.}
\]
\end{proof}

\begin{proof}[Proof of Proposition \ref{pro:ergo}]
For the first claim, let $Q_{x,y}^n$ be the the optimal coupled measure between $\delta_{x}\Phat^n $ and $\delta_{y}\Phat^n $ for
any $x,y\in\Omega$.
Then,
\begin{equation}\label{eq:n_bd}
\begin{split}
|\delta_x \Phat^n f -\pihat f|&=|\delta_{x}\Phat^n f-\pihat \Phat^n f|\\
&\leq \int |\delta_{x}\Phat^n f-\delta_y\Phat^n f|\pihat(dy)\\
&\leq \int \int Q^n_{x,y}(dx^{\prime},dy^{\prime})|f(x^{\prime})-f(y^{\prime})|\pihat(dy)\\
&\leq \int \|\delta_{x}\Phat^n - \delta_y\Phat^n\|_V \pihat(dy)\\
&\leq \rhohat^n \int d_V(x,y)\pihat(dy)\leq \rhohat^n(V(x)+\pihat V).
\end{split}\end{equation}
For the second claim, note that for any $f$ with $\pihat f=0$, we have
\begin{align*}
\E_{\pihat}[(\fhat_M-\pihat f)^2]&=\frac{1}{M^2}\E_{\pihat} \left[\sum_{j,k=1}^M f(\Xhat_j) f(\Xhat_{k})\right]\\
&\leq \frac{2}{M^2}\E_{\pihat} \left[\sum_{j=1}^M |f(\Xhat_j)| \sum_{k=0}^\infty |f(\Xhat_{j+k})|\right]\\
&= \frac{2}{M}\E_{\pihat} \left[|f(\Xhat_0)|\E\Big[\sum_{k=0}^\infty |f(\Xhat_{k})|\Big|\Xhat_0\Big]\right]\\
&\leq \frac{2}{M}\E_{\pihat} \left[|f(\Xhat_0)| \sum_{k=0}^\infty \rhohat^k(V(\Xhat_0)+\pihat V)\right] \mbox{ from \eqref{eq:n_bd}}\\
&\leq \frac{2}{ (1-\rhohat)M}\E_{\pihat}\Big[|f(\Xhat_0)|(V(\Xhat_0)+\pihat V)\Big].
\end{align*}
\end{proof}
\section{Proof for the spectral gap framework}
\begin{proof}[Proof of Theorem \ref{th:app_gap}]
We will write $\kappa=\kappa(P)$ for short. \\
\noindent\textbf{For Claim 1,} first note that
\begin{equation} \label{eq:th3_bd1}
\begin{split}
\var_\pi (\Phat f)=&\frac12\int\left(\Phat f(x)-\Phat f(y)\right)^2\pi(dx)\pi(dy)\\
=&\frac12\int\left(Pf(x)-Pf(y)+(P-\Phat)f(y)-(P-\Phat) f(x)\right)^2\pi(dx)\pi(dy)\\
\leq& \left(\frac12+\frac{a\kappa}{2}\right)\int\left(Pf(x)-Pf(y)\right)^2\pi(dx)\pi(dy)\\
&+\left(\frac12+\frac{1}{2a\kappa}\right)\int\left((P-\Phat) f(x)-(P-\Phat) f(y)\right)^2\pi(dx)\pi(dy)\\
\leq& (1+a \kappa) \var_\pi(P f)+\left(2+\frac{2}{a\kappa}\right)\int((P-\Phat)f(x))^2\pi(dx).
\end{split}\end{equation}
Next, note that
\begin{equation}\label{eq:th3_bd2}
(1+a\kappa) \var_\pi(P f(X))\leq (1+a \kappa) (1-\kappa)\var_\pi f
\leq  (1-(1-a) \kappa) \var_\pi f.
\end{equation}
Let $\Delta f(x)=f(x)-\pi f$. Then,
\begin{equation}\label{eq:th3_bd3}
\E_\pi\left[\left((P-\Phat) f(X)\right)^2\right]=\E_\pi\left[\left((P-\Phat)\Delta f(X)\right)^2\right]\leq \epsilon^2 \|\Delta f\|_\pi^2=\epsilon^2 \var_\pi f.
\end{equation}
Plugging the bounds \eqref{eq:th3_bd2} and \eqref{eq:th3_bd3} in \eqref{eq:th3_bd1}, we have
\[
\var_\pi (\Phat f)\leq \left(1-(1-a) \kappa+\frac{(2a\kappa+2)\epsilon^2}{a\kappa}\right) \var_\pi f. 
\]
This further implies that
\[
\var_\pi (\Phat^n f)\leq \left(1-(1-a) \kappa+\frac{C\epsilon^2}{a}\right)^n \var_\pi f.
\]

\noindent\textbf{For Claim 2,} first note that
\begin{equation*}\begin{split}
\left|\int \Phat f(x)\pi(dx) - \pi f \right| &= \left|\int (\Phat-P)f(x)\pi(dx) \right|\\
&=\left|\int (\Phat-P)\Delta f(x)\pi(dx) \right| \mbox{  recall that $\Delta f(x)=f(x)-\pi f$}\\
&\leq \left(\int \left((\Phat-P)\Delta f(x)\right)^2\pi(dx) \right)^{1/2}\\
&=\|(P-\Phat) \Delta f\|_\pi\leq \epsilon\sqrt{\var_\pi f}.
\end{split}\end{equation*} 
Then, because $\int P g(x)\pi(dx)=\int g(x)\pi(dx)$ for any function $g$,
\begin{equation*}\begin{split}
\left|\int (\Phat^{n+1}-\Phat^n)f(x)\pi(dx)\right|&=\left|\int (\Phat-P)(\Phat^nf)(x)\pi(dx)\right|\\
&\leq \epsilon \sqrt{\var_\pi (\Phat^n f)}\leq \epsilon (1-\kappahat)^{n/2}\sqrt{ \var_\pi f},
\end{split}\end{equation*}
from Claim 1.
Let $\mu_n=\pi \Phat^n$ and
 $f=\text{sgn}(\mu_{n+1}-\mu_n)$. Because $\var_\pi f\leq \E_\pi |f|^2=1$, 
\[
\|\mu_{n+1}-\mu_n\|_{TV}=\left|\int (\Phat^{n+1}-\Phat^n)f(x)\pi(dx)\right|\leq \epsilon (1-\kappahat)^{n/2}\sqrt{ \var_\pi f} \leq \epsilon(1-\kappahat)^{n/2}.
\]
Thus, $\mu_n$ is a Cauchy sequence under the total variation metric, which implies that the sequence has a limit $\pihat$ and 
\[
\left|(\pihat-\mu_n)f\right|\leq \sum_{k=n}^\infty \left|(\mu_{k+1}-\mu_k)f\right|
\leq \frac{\epsilon(1-\kappahat)^{n/2}}{1-(1-\kappahat)^{1/2}}\sqrt{ \var_\pi f}
\]
When letting $n=0$, we have
\begin{equation}\label{eq:bd_pipi}
|\pihat f -\pi f|\leq \frac{\epsilon}{1-(1-\kappahat)^{1/2}} \sqrt{\var_\pi f}
\end{equation}
Consider $f=\pihat/\pi$. $D_{\chi^2}(\pihat\|\pi)=|\pihat f -\pi f|=\var_\pi f$. Combine this with \eqref{eq:bd_pipi}, we have
\[
D_{\chi^2}(\pihat\|\pi)\leq \frac{\epsilon^2}{(1-(1-\kappahat)^{1/2})^2}.
\]

\end{proof}

\begin{proof}[Proof of Proposition \ref{prop:app_gap}]
For the first claim, we note 
\begin{align*}
|\nu\Phat^n f-\pi \Phat^n f|&\leq \int \frac{\nu(x)}{\pi(x)}\pi(dx)|\Phat^n f(x)-\pi \Phat^n f|\\
&\leq \sqrt{\int \Big(\frac{\nu(x)}{\pi(x)}\Big)^2 \pi(dx)} \sqrt{\int |\Phat^n f(x)-\pi \Phat^nf |^2\pi(dx)}\\
&\leq \sqrt{D_{\chi^2}(\nu\|\pi)+1}\times (1- \hat{\kappa})^{n/2}\sqrt{\var_\pi f} \mbox{ by Theorem \ref{th:app_gap}.}
\end{align*}
Meanwhile,
\[
|(\pihat-\pi \Phat^n)f| \leq C\epsilon(1-\hat{\kappa})^{n/2}\sqrt{\var_\pi f} \mbox{ by Theorem \ref{th:app_gap}.}
\]
By triangular inequality,
\begin{equation*}\begin{split}
|\nu\Phat^n f-\pihat f|&\leq |\nu\Phat^n f-\pi \Phat^n f|+|\pi \Phat^nf-\pihat f|\\
&\leq (1-\kappahat)^{n/2}\sqrt{\var_\pi f}\left(\sqrt{D_{\chi^2}(\nu\|\pi)+1}+C\epsilon\right).
\end{split}\end{equation*}

For the second part, we first note that for any $f$ with $\pihat f=0$, we have
\begin{align*}
\E_{\pihat}[(\fhat_M-\pihat f)^2]&=\frac{1}{M^2}\E_{\pihat}\left[\sum_{j,k=1}^M f(\Xhat_j) f(\Xhat_{k})\right]\\
&\leq \frac{2}{M^2}\E_{\pihat}\left[\sum_{j=1}^M |f(\Xhat_j)| \sum_{k=0}^\infty |f(\Xhat_{j+k})|\right]\\
&= \frac{2}{M}\sum_{k=0}^\infty \E_{\pihat}\left[|f(\Xhat_0)||f(\Xhat_{k})|\right]\\
&\leq \frac{2}{M}\sqrt{\E_{\pihat}[f(\Xhat_0)^2]} \sum_{k=0}^\infty \sqrt{\E_{\pihat}[(\Phat^k f(\Xhat_0))^2]}.
\end{align*}
Next,
\begin{align*}
\var_{\pihat}(\Phat^k f)&=\pihat(\Phat^k f)^2\\
&\leq \pihat(\Phat^k f-\pi\Phat^k f)^2\\
&\leq \pi(\Phat^k f-\pi\Phat^k f)^2+\frac{\epsilon}{1-(1-\kappahat)^{1/2}} \sqrt{\var_\pi[(\Phat^k f-\pi\Phat^k f)^2]} \mbox{ by \eqref{eq:bd_pipi}}\\
&\leq  (1-\kappahat)^k \var_\pi (f)+\frac{\epsilon}{1-(1-\kappahat)^{1/2}} \sqrt{\var_\pi[(\Phat^k f-\pi\Phat^k f)^2]} \mbox{ by Theorem \ref{th:app_gap}.}
\end{align*}
Because $\sup_x|f(x)|\leq C$ for some $C\in(0,\infty)$, $\sup_x|\Phat^k f(x)|\leq C$ and $\sup_x|(\Phat^k-\pi \Phat^k)f(x)|\leq 2C$. Then,
\[
\pi(\Phat^k f-\pi \Phat^k f)^4\leq 4C^2\pi(\Phat^k f-\pi \Phat^k f)^2\leq 4C^2 (1-\kappahat)^k \var_\pi f \mbox{ by Theorem \ref{th:app_gap}.}
\]
Thus,
\[
\var_{\pihat}(\Phat^k f) \leq  (1-\kappahat)^k \var_\pi (f)+\frac{\epsilon}{1-(1-\kappahat)^{1/2}} 2C (1-\kappahat)^{\frac{k}{2}}\sqrt{\var_\pi f}
\]
and we can further find a constant $C'$ such that
\[
\E_{\pihat}[(\fhat_M-\pihat f)^2] \leq \frac{C'}{M(1-(1-\kappahat)^{1/4})}\sqrt{\var_{\pihat}f~\var_\pi f}.
\]
\end{proof}

\begin{proof}[Proof of Theorem \ref{th:gap_ratio}]
To simplify the notation, let $\kappa$ denote the spectral gap of $P$ and $\kappahat$ denote the spectral gap of $\Phat$.
By the definition of spectral gap, i.e., \eqref{eq:gap}, we have
\begin{align*}
\kappahat&=\min_f \frac{ \langle f, (I-\Phat^2) f\rangle_{\pihat}}{\var_{\pihat} f}.
\end{align*}
First, note that 
\begin{equation}\label{eq:var}
\begin{split}
\var_{\pihat} f&=\E_{\pihat} [(f-\pihat f)^2]\\
&\leq \E_{\pihat} [(f-\pi f)^2]\\
&\leq  (1+\epsilon)\E_{\pi}[(f-\pi f)^2] \\ 
&=(1+\epsilon)\var_\pi f. 
\end{split}\end{equation}
We next establish two useful bounds:
\begin{equation}\label{eq:b1}
\begin{split}
|\langle f, (I-P^2)f\rangle_{\pihat}-\langle f, (I-P^2)f\rangle_{\pi}|&\leq \int |f(x)(I-P^2)f(x) |\epsilon \pi(dx) \\
&\leq \epsilon\|f\|_\pi \|(I-P^2)f\|_\pi
\leq \epsilon\|f\|_\pi^2,
\end{split}\end{equation}
and
\begin{equation}\label{eq:b2}
\begin{split}
|\langle f, (\Phat^2-P^2)f\rangle_{\pihat}|&\leq \|f\|_{\pihat} \|(\Phat^2-P^2)f\|_{\pihat}\\
&\leq (1+\epsilon)^2\|f\|_\pi \|(\Phat^2-P^2)f\|_{\pi} \\ 
&\leq (1+\epsilon)^2\|f\|_\pi (2\|(\Phat-P)Pf\|_{\pi}+\|(\Phat-P)^2f\|_{\pi}) \\ 
&\leq (1+\epsilon)^2\|f\|_\pi (2\epsilon\|Pf\|_{\pi}+\epsilon\|(\Phat-P)f\|_{\pi}) \\
&\leq (1+\epsilon)^2\|f\|_\pi (2\epsilon\|f\|_{\pi}+\epsilon^2\|f\|_{\pi}) \\ 
&\leq 3(1+\epsilon)^2\epsilon\|f\|_\pi^2\\
&\leq C\epsilon\|f\|_\pi^2.
\end{split}\end{equation}
Then,
\begin{equation}\label{eq:Di}
\begin{split}
|\langle f, (I-\Phat^2) f\rangle_{\pihat}|&\geq|\langle f, (I-P^2) f\rangle_{\pihat}| - |\langle f, (\Phat^2-P^2)f\rangle_{\pihat}| \mbox{ by triangular inequality}\\
&\geq |\langle f, (I-P^2) f\rangle_{\pihat}|-C\epsilon\|f\|_\pi^2 ~\mbox{ by the bound in \eqref{eq:b2}}\\
&\geq |\langle f, (I-P^2) f\rangle_{\pi}| - |\langle f, (I-P^2) f\rangle_{\pi}-\langle f, (I-P^2) f\rangle_{\pihat}|-C\epsilon\|f\|_\pi^2\\
&\geq  |\langle f, (I-P^2) f\rangle_{\pi}|-C\epsilon\|f\|_\pi^2 -C\epsilon\|f\|_\pi^2 ~\mbox{ by the bound in \eqref{eq:b1}}\\
&\geq \langle f, (I-P^2) f\rangle_{\pi} - C\epsilon \var_\pi f.
\end{split}\end{equation}
Combining \eqref{eq:var} and \eqref{eq:Di}, we have $\kappahat \geq \kappa -C\epsilon. $

\end{proof}

\begin{proof}[Proof of Proposition \ref{pro:gapproof}]
For the first claim, 
\begin{align*}
|\nu\Phat^n f-\pihat f|&=|\nu\Phat^n f-\pihat \Phat^n f|\\
&\leq \int \frac{\nu(x)}{\pihat(x)}\pihat(x)|\Phat^n f(x)-\pihat \Phat^n f|dx\\
&\leq \sqrt{\int \Big(\frac{\nu(x)}{\pihat(x)}\Big)^2 \pihat(x)dx} \sqrt{\int |\Phat^n f(x)-\pihat \Phat^nf |^2\pihat(x)dx}\\
&\leq \sqrt{D_{\chi^2}(\nu\|\pihat)+1}\times (1- \hat{\kappa})^{n/2}\sqrt{\var_{\pihat} f}. 
\end{align*}

For the second part, we first note that for any $f$ with $\pihat f=0$, we have
\begin{equation*}\begin{split}
\E_{\pihat}[(\fhat_M-\pihat f)^2]&=\frac{1}{M^2}\E_{\pihat}\left[\sum_{j,k=1}^M f(\Xhat_j) f(\Xhat_{k})\right]\\
&\leq \frac{2}{M^2}\E_{\pihat}\left[\sum_{j=1}^M |f(\Xhat_j)| \sum_{k=0}^\infty |f(\Xhat_{j+k})|\right]\\
&= \frac{2}{M}\sum_{k=0}^\infty \E_{\pihat}\left[|f(\Xhat_0)| |f(\Xhat_{k})|\right]\\
&\leq \frac{2}{M}\sqrt{\E_{\pihat}[f(\Xhat_0)^2]} \sum_{k=0}^\infty \sqrt{\E_{\pihat}[(\Phat^k f(\Xhat_0))^2]}\\
&\leq  \frac{2}{M}\sqrt{\E_{\pihat}[f(\Xhat_0)^2]} \sum_{k=0}^\infty (1-\kappahat)^{k/2}\sqrt{\var_{\pihat} f}\\
&=\frac{2}{M(1-(1-\kappahat)^{1/2})}\var_{\pihat} f
\end{split}\end{equation*}
\end{proof}

\begin{proof}[Proof of Proposition \ref{prop:pert}]
Let $Q_x$ be the optimal coupled measure between $\delta_x P$ and $\delta_x\Phat$. Then
\begin{align*}
|(\delta_xP-\delta_x\Phat) f|&\leq \int Q_x(dx',dy') |f(x')-f(y')|\\
&= \int Q_x(dx',dy') (|f(x')-f(y')|)1_{x'\neq y'}\leq \epsilon  V(x). 
\end{align*}
Next, 
\begin{align*}
\|(P-\hat P)f\|_\pi^2=&\int \pi(dx)|(\delta_xP-\delta_x\Phat) f|^2\\
\leq& \int \pi (dx) \left(\int  Q_x(dx',dy') (|f(x')-f(y')|)1_{x'\neq y'}\right)^2\\
\leq& \left(\int \pi (dx)Q_x(dx',dy') \left(2f(x')^2+2f(y')^2\right)\right)\left(\int \pi (dx) Q_x(dx',dy') 1_{x'\neq y'}\right)\\
\leq& 2 \left(\langle \pi P, f^2\rangle+\langle \pi \Phat, f^2\rangle\right) \left(\epsilon \int \pi(dx)V(x) \right)\\
\leq& 2\epsilon\left(\langle \pi, f^2\rangle+a\langle \pihat \Phat, f^2\rangle\right)(\pi V)\\
\leq& 2\epsilon\left(1+a^2\right)\|f\|_{\pi}^2\|V\|_\pi.
\end{align*}
\end{proof}

\section{Verification for Metropolis Hasting MCMC }

We first present an auxiliary lemma that will be used in our subsequent development.

\begin{lem}
\label{lem:l2pi2hat}
Suppose $P$ is irreducible and reversible with invariant measure $\pi$. Then, 
$\|P\|_\pi\leq 1$.
\end{lem}
\begin{proof}
We first note that
\begin{equation*}\begin{split}
&\int f(x)^2\pi(dx) - \int \pi(dx)f(x)f(y)P^2(x,dy)\\
=&\frac{1}{2}\int \pi(dx)f(x)^2P^2(x,dy)+\frac{1}{2}\int\pi(dx)f(y)^2P^2(x,dy) -\int \pi(dx)f(x)f(y)P^2(x,dy)\\
=&\frac{1}{2}\int \pi(dx)(f(x)-f(y))^2P^2(x,dy)\geq 0.
\end{split}\end{equation*}
Thus, 
\[
\int f(x)^2\pi(dx)\geq \int \pi(dx) f(x)f(y)P^2(x,dy)= \|Pf\|_\pi^2. 
\]
\end{proof}

\begin{proof}[Proof of Lemma \ref{lem:MHtransition}]
For any density of form $\mu(x)=\nu(x)s(x)$, we have
\begin{align*}
| \mu (P-\Phat) f|&\leq \int \mu(x) |\alpha(x)-\hat{\alpha}(x)| |f(x)|dx+\int \mu(x) |\beta(x,x')-\hat{\beta}(x,x')| |f(x')|dx'dx\\
&\leq C\epsilon \int \mu(x)  |f(x)|dx+C\epsilon \int \mu(x) \beta(x,x') |f(x')|dx'dx\\
&\leq C\epsilon \int \mu(x)  |f(x)|dx+C\epsilon \int \mu(x) P(x,x') |f(x')|dx'dx\\
&=C\epsilon \int \nu(x)s(x)  |f(x)|dx+C\epsilon \int \nu(x)s(x) P(x,x') |f(x')|dx'dx\\
&\leq C\epsilon \| s \|_\nu \|f\|_\nu+C\epsilon \| s \|_\nu \|P f\|_\nu\\
&\leq 2C\epsilon \| s \|_\nu \|f\|_\nu \mbox{ by Lemma \ref{lem:l2pi2hat}.}
\end{align*}
Next, take $\mu\propto |(P-\Phat)f|\nu$, we have
\[\|(P-\Phat)f\|_{\nu}^2 \leq 2C\epsilon\|(P- \Phat)f\|_{\nu}\|f\|_{\nu},\]
which further implies that there is a $C_1$, so that $\|P-\Phat\|_\nu\leq C_1\epsilon.$
\end{proof}

\begin{proof}[Proof of Proposition \ref{prop:l2MWN}]

%
%
We denote the acceptance probabilities for the original process and perturbed process as 
\[b(x,x')=\frac{\pi(x')}{\pi(x) }\wedge 1 \mbox{ and } \hat b(x,x')=\frac{\hat \pi(x')}{\hat \pi(x) }\wedge 1\]
respectively.
Since for any positive numbers $a,b,c,d$, 
\[
\min \{a/b, c/d\}\leq \frac{a\wedge c}{b\wedge d}\leq \max \{a/b, c/d\},
\]
and since $\exp(-C\epsilon)\leq \pi(x)/\hat\pi(x)\leq \exp(C\epsilon)$, we have
\[\exp(-2C\epsilon)b(x,x')<\hat b(x,x')<\exp(2C\epsilon)b(x,x').\]
Using the fact that $\beta(x,x')=R(x,x')b(x,x')$ and $\hat\beta(x,x')=R(x,x')\hat b(x,x')$, we have
\[\exp(-2C\epsilon)\beta(x,x')<\hat \beta(x,x')<\exp(2C\epsilon)\beta(x,x').\]
In addition, for $\alpha(x)=\int R(x,x')(1-b(x,x'))dx'$ and $\hat \alpha(x)=\int R(x,x')(1-\hat b(x,x'))dx'$, 
\[|\alpha(x)-\hat\alpha(x)| \leq \int R(x,x')|b(x,x')-\hat b(x,x')|dx' \leq C\epsilon.\]
By Lemma \ref{lem:MHtransition}, we can find a $C_1$ so that
\[
\|P_{RWM} - \Phat_{RWM}\|_{\pi}\leq C_1\epsilon.
\]
\end{proof}

\begin{proof}[Proof of Proposition \ref{prop:l2MALA}]
Note that
\[
R(x,x')=\frac{1}{(4\pi h)^{d/2}}\exp\left(-\frac1{4h}\|x'-x-\nabla \log \pi(x)h \|^2\right),
\]
and
\[
\Rhat(x,x')=\frac{1}{(4\pi h)^{d/2}}\exp\left(-\frac1{4h}\|x'-x-\nabla \log \hat{\pi}(x)h \|^2\right).
\]
As $|\nabla\log  \pihat(x)- \nabla\log \pi(x)|\leq C\epsilon$
and the support is bounded, we can enlarge the value of $C$ so that
\[
(1-C\epsilon)R(x,x')\leq \Rhat(x,x')\leq (1+C\epsilon) R(x,x'),
\]
Let the acceptance probability be 
\begin{equation*}\begin{split}
b(x,x')=&\frac{\pi(x')\exp\left(-\frac1{4h}\|x-x'+\nabla \log \pi(x')h \|^2\right)}{\pi(x) \exp\left(-\frac1{4h}\|x'-x+\nabla \log \pi(x)h\|^2\right)}\wedge 1\\
=&\left\{\exp\left(\log \pi(x')-\log \pi(x)-\frac1{2}\langle x-x',\nabla \log \pi(x')-\nabla \log\pi(x)\rangle\right)\right.\\
&\left.\times \exp\left(\frac{h}{4} \left[\|\nabla \log \pi(x')\|^2-\|\nabla \log \pi(x)\|^2\right]\right)\right\}\wedge 1
\end{split}\end{equation*}
Similarly, we define
\begin{equation*}\begin{split}
\hat b(x,x') =& \frac{\hat\pi(x')\exp\left(-\frac1{4h}\|x-x'+\nabla \log \hat \pi(x')h \|^2\right)}{\hat\pi(x) \exp\left(-\frac1{4h}\| x'-x+\nabla \log \hat\pi(x)h\|^2\right)}\wedge 1 \\
=&\left\{\exp\left(\log \hat\pi(x')-\log \hat\pi(x)-\frac1{2}\langle x-x',\nabla \log \hat\pi(x')-\nabla \log\hat\pi(x)\rangle\right)\right.\\
&\left.\times \exp\left(\frac{h}{4} \left[\|\nabla \log \hat\pi(x')\|^2-\|\nabla \log \hat\pi(x)\|^2\right]\right)\right\}\wedge 1
\end{split}\end{equation*}
Since $|\log \pi(x)-\log \pihat_k(x)|\leq C\epsilon, \|\nabla\log \pi(x)-\nabla\log \pihat_k(x)\|\leq C\epsilon$, and the support is bounded,
we can further enlarge $C$, such that 
\[
(1-C\epsilon)b(x,x')\leq \hat b(x,x')\leq (1+C\epsilon) b(x,x').
\]

Lastly, for $\beta(x,x')=R(x,x')b(x,x')$ and $\hat\beta(x,x')=\hat R(x,x')\hat b(x,x')$,
\[(1-C\epsilon)\beta(x,x')\leq \hat \beta(x,x') \leq (1+C\epsilon)\beta(x,x').\]
In addition, for $\alpha(x)=\int R(x,x')(1-b(x,x'))dx'$ and $\hat \alpha(x)=\int \hat R(x,x')(1-\hat b(x,x'))dx'$,
\begin{equation*}\begin{split}
|\alpha(x) - \hat \alpha(x)| \leq & \int |R(x,x')-\hat R(x,x')|(1-b(x,x'))dx'\\
& + \int \hat R(x, x')|b(x,x')-\hat b(x,x')|dx'\\
\leq& C\epsilon \int R(x,x')dx' + C\epsilon \int \hat R(x,x')dx'
=2C\epsilon
\end{split}\end{equation*}
By Lemma \ref{lem:MHtransition}, we have a constant $C_1$ so that
\[
\|P_{MALA}-\Phat_{MALA}\|_{\pi}\leq C_1\epsilon.
\]

\end{proof}

\begin{proof}[Proof of Proposition \ref{prop:MALA_TV}]
The transition kernel of MALA takes the form
\[P(x,y)=\alpha(x)\delta_x(y)+\beta(x,y)\]
where 
\[
\beta(x,y)=\frac{\pi(y)}{\pi(x)}q(y,x)\wedge q(x,y)=a(x,y)\wedge q(x,y),
\]
with
\[
q(x,y)=\frac{1}{(2\pi h)^{d/2}}\exp\left(-\frac1{4h}\|y-x-\nabla \log \pi(x)h \|^2\right), a(x,y)=\frac{\pi(y)}{\pi(x)}q(y,x),
\]
and $\alpha(x)=1-\int \beta(x,dy)$.
Similarly, we can write $\Phat(x,y)=\alphahat(x)\delta_x(y) + \betahat(x,y)$ when using the perturbed target density $\pihat$.

We prove the proposition by showing that 
\begin{equation} \label{eq:bd_mala}
\int |q(x,y)-\qhat(x,y)|dy\leq C\epsilon \mbox{ and }
\int |a(x,y)-\widehat{a}(x,y)|dy\leq C\epsilon\exp(\delta x^2).
\end{equation}
In particular, note that $a\wedge q-\widehat{a}\wedge \widehat{q}\in \{a-\hat{a}, q-\hat{q}, a-\hat{q},\hat{a}-q\}$,
which further implies that
\[
|a\wedge q-\widehat{a}\wedge \widehat{q}|\leq |a-\widehat{a}|+|q-\widehat{q}|. 
\]
Thus, if the bounds in \eqref{eq:bd_mala} hold, then
\begin{align*}
\|\delta_xP-\delta_x\Phat\|_{TV}&=\int |\beta(x,y)-\widehat{\beta}(x,y)|dy+|\alpha(x)-\widehat{\alpha}(x)|\\
&\leq 2\int |\beta(x,y)-\widehat{\beta}(x,y)|dy\\
&\leq 2\int |a(x,y)-\widehat{a}(x,y)|dy+2\int |q(x,y)-\widehat{q}(x,y)|dy\\
&\leq 2C\epsilon (1+\exp(\delta x^2)). 
\end{align*}
In order to obtain the first part of \eqref{eq:bd_mala}, note that by intermediate value theorem, $|\exp(a)-\exp(b)|\leq |\exp(a)+\exp(b)||a-b|$ holds for any $a,b$, so we can bound
\begin{align}
\notag
|q(x,y)-\qhat(x,y)|\leq& \frac1{4}|q(x,y)+\qhat(x,y)|
\|\nabla\log  \pi(x)-\nabla\log \pihat(x)\|\\
\notag
&\left(\|y-x-h\nabla\log  \pi(x)\|+\|y-x-h\nabla\log  \pihat(x)\|\right)\\
\label{tmp:4}
\leq& \frac{C\epsilon}4|q(x,y)+\qhat(x,y)|\left(\|y-x-h\nabla\log  \pi(x)\|+\|y-x-h\nabla\log  \pihat(x)\|\right).
\end{align}
Note that $q(x,y)$ is the proposal density of $y$. Thus,
\begin{align*}
&\int q(x,y)\left(\|y-x-h\nabla\log  \pi(x)\|+\|y-x-h\nabla\log  \pihat(x)\|\right) dy\\
&\leq \int q(x,y)\left(2\|y-x-h\nabla\log  \pi(x)\|+Ch\epsilon\right)dy\\\
&\leq Ch\epsilon +2\sqrt{\int q(x,y)\|y-x-h\nabla\log  \pi(x)\|^2dy}=Ch\epsilon+2\sqrt{2h d}. 
\end{align*}
Similarly,
\[
\int q(x,y)\left(\|y-x-h\nabla\log  \pi(x)\|+\|y-x-h\nabla\log  \pihat(x)\|\right)dy\leq Ch\epsilon+2\sqrt{2h d}.
\]
Therefore, we find use \eqref{tmp:4} and find a larger $C$ so that 
\begin{align*}
\int |q(x,y)-\qhat(x,y)|dy&\leq  C\epsilon. 
\end{align*}
To handle the second part of \eqref{eq:bd_mala}, we use
$|\exp(a)-\exp(b)|\leq |\exp(a)+\exp(b)||a-b|$ again and find
\begin{align*}
&|a(x,y)-\widehat{a}(x,y)|\\
=&\left|\frac{\pi(y)}{\pi(x)}q(y,x)-\frac{\pihat(y)}{\pihat(x)}\qhat(y,x)\right|\\
\leq& \frac14\left|\frac{\pi(y)}{\pi(x)}q(y,x)+\frac{\pihat(y)}{\pihat(x)}\qhat(y,x)\right|\bigg(|\log  \pi(x)-\log \pihat(x)|+|\log  \pi(y)-\log \pihat(y)|\\
&+\|\nabla\log  \pi(y)-\nabla\log \pihat(y)\|\left(\|y+h\nabla\log  \pi(y)-x\|+\|y+h\nabla\log  \pihat(y)-x\|\right)\bigg)\\
\leq& \frac{C\epsilon}{4}\left|\frac{\pi(y)}{\pi(x)}q(y,x)+\frac{\pihat(y)}{\pihat(x)}\qhat(y,x)\right|\left(2+\left(\|y+h\nabla\log  \pi(y)-x\|+\|y+h\nabla\log  \pihat(y)-x\|\right)\right).
\end{align*}
Note that the first part can be bounded by 
\begin{align*}
&\int \frac{\pi(y)}{\pi(x)}q(y,x)\left(C+(\|y+h\nabla\log  \pi(y)-x\|+\|y+h\nabla\log  \pihat(y)-x\|\right)dx\\
\leq& \int \frac{\pi(y)}{\pi(x)}q(y,x)\left(C+h\epsilon+2\|y+h\nabla\log  \pi(y)-x\|\right)dx\\
=& \int \frac{\pi(y)}{(2\pi h)^{d/4}\pi(x)}\sqrt{q(y,x)}\,\cdot\, (2\pi h)^{d/4}\sqrt{q(y,x)}\left(C+h\epsilon+2\|y+h\nabla\log  \pi(y)-x\|\right)dx\\
\end{align*}
For $\tfrac{\pi(y)}{\pi(x)}(2\pi h)^{-d/4}\sqrt{q(y,x)}$, we can bound it by 
\begin{equation*}\begin{split}
&\frac{1}{(2\pi h)^{d/4}}\frac{\pi(y)}{\pi(x)}\sqrt{q(y,x)}\\
=&\frac{1}{(2\pi h)^{d/2}}\exp\left(\log\pi(y)-\log\pi(x)-\frac1{8h}\|y+h\nabla\log  \pi(y)-x\|^2\right)\\
\leq& \frac{1}{(2\pi h)^{d/2}}\exp\left(\langle\nabla\log\pi(w), y-x\rangle-\frac1{8h}\|y-x\|^2-\frac{1}{4}\langle\nabla\log  \pi(y), y-x\rangle\right)
\mbox{ for some $w$}\\
\leq&\frac{1}{(2\pi h)^{d/2}} \exp\left(\frac{5L_\pi}{4}\|y-x\|(\|x\|+\|y-x\|+C)-\frac1{8h}\|y-x\|^2\right)\mbox{ by Lipschitzness of  $\nabla \log \pi$}\\
\leq&\frac{1}{(2\pi h)^{d/2}} \exp\left(\left(\frac{5L_\pi}{16\delta} + \frac{5L_{\pi}}{4} - \frac{1}{8h}\right)\|y-x\|^2 + \delta \|x\|^2+10L_\pi^2 C^2\right)\\
\leq& \frac{1}{(2\pi h)^{d/2}} \exp\left( - \frac{1}{16h}\|y-x\|^2 + \delta \|x\|^2+10L_\pi^2 C^2\right) \mbox{ as $h<(\frac{5L_{\pi}}{\delta} + 20L_{\pi})^{-1}$.}
\end{split}\end{equation*}
For $(2\pi h)^{d/4}\sqrt{q(y,x)}\left(C+h\epsilon+2\|y+h\nabla\log  \pi(y)-x\|\right)$, first note that we can find a larger $C$ so that
\begin{align*}
&(2\pi h)^{d/4}\sqrt{q(y,x)}\|y+h\nabla\log  \pi(y)-x\|\\
=&\exp\left(-\frac1{8h}\|y+h\nabla\log  \pi(y)-x\|^2\right)\|y+h\nabla\log  \pi(y)-x\|\leq  C. 
\end{align*}
Combining these two upper bound, we can find a $C_1$ so that 
\[
\int \frac{\pi(y)}{\pi(x)}q(y,x)\left(C+h\epsilon+2\|y+h\nabla\log  \pi(y)-x\|\right)dx\leq C_1\exp(\delta \|x\|^2). 
\]
Similarly, we can show that
\begin{align*}
&\int \frac{\pihat(y)}{\pihat(x)}\qhat(y,dx)\left(C+(\|y+h\nabla\log  \pi(y)-x\|+\|y+h\nabla\log  \pihat(y)-x\|\right)\\
\leq& \int \frac{\pihat(y)}{\pihat(x)}\qhat(y,dx)\left(C+h\epsilon+2\|y+h\nabla\log  \pihat(y)-x\|\right)
\leq C\exp(\delta \|x\|^2).
\end{align*}
Thus, $\int |a(x,y)-\widehat{a}(x,y)|dy\leq C\epsilon\exp(\delta x^2)$ for some $C$. This concludes the proof of \eqref{eq:bd_mala} and our claim. 
\end{proof}

\section{Verification for the parallel tempering algorithm}
\begin{proof}[Proof of Lemma \ref{lem:l2comp}]
For claim 1), note that for any $\|f\|_\nu\leq 1$,
\begin{equation*}\begin{split}
\|RSf-\Rhat\Shat f\|_\nu&\leq \|R(S-\Shat)f\|_\nu+\|(R-\Rhat)\Shat f\|_\nu\\
&\leq \|(S-\Shat)f\|_\nu+C\epsilon\|\Shat  f\|_\nu \mbox{ by Lemma \ref{lem:l2pi2hat}}\\
&\leq C\epsilon \|f\|_\nu+C\epsilon\| (\Shat - S) f\|_\nu + C\epsilon\|S f\|_{\nu} \\
&\leq  (2C+C^2\epsilon)\epsilon \|f\|_\nu \mbox{ by Lemma \ref{lem:l2pi2hat}.}
\end{split}\end{equation*}
For  claim 2), we first note that 
\[
R_1\otimes R_2=(R_1\otimes I)(I\otimes R_2)
\]
We will show that
\[
\|(R_1\otimes I)-(\Rhat_1\otimes I)\|_{\nu}=\|((R_1-\Rhat_1)\otimes I)\|_{\nu}\leq C\epsilon.
\]
For any $f(x,y)$, define
\[
g(x,y):=((R_1-\Rhat_1)\otimes I)f(x,y)
\]
Then for each fixed $y$, since $\|R_1-\Rhat_1\|_{\nu_1}\leq C \epsilon$, 
\[
\int g(x,y)^2 \nu_1(x)dx\leq  C^2\epsilon^2 \int f(x,y)^2 \nu_1(x)dx 
\]
Thus,
\[
\|g\|^2_{\nu}=\int g(x,y)^2 \nu_1(x)\nu_2(y)dxdy\leq  C^2\epsilon^2 \int \int f(x,y)^2 \nu_1(x)\nu_2(y)dxdy=C^2\epsilon^2 \|f\|^2_\nu. 
\]
Similarly, we can show that
\[
\|(I\otimes R_2)-(I\otimes \Rhat_2)\|_\nu=\|I\otimes (R_2-\Rhat_2)\|_\nu\leq C\epsilon.
\]
From claim 1), we can find a $C'$ so that
\[
\|R_1\otimes R_2\|_{\nu}=\|(R_1\otimes I)(I\otimes R_2)\|_{\nu} \leq C'\epsilon.
\]
For claim 3), by triangular inequality, we have
\[\|U - \Uhat\|_\nu \leq \frac{1}{n}\sum_{i=1}^{n}\|S_i-\Shat_i\|_{\nu} \leq C \epsilon.\]
\end{proof}

\begin{proof}[Proof of Lemma \ref{lem:REl2}]
Denote $x'=S(x)$. For any density of form $\mu(x)=s(x)\nu(x)$, we have
\begin{align*}
|\mu (Q-\Qhat) f|&\leq \int \mu(x) |a(x,x')-\widehat{a}(x,x')| |f(x)|dx+\int \mu(x) |a(x,x')-\widehat{a}(x,x')| |f(x')|dx\\
&\leq C\epsilon \int \mu(x) |f(x)|dx+C\epsilon \int \mu(x) a(x,x') |f(x')|dx\\
&\leq C\epsilon \int s(x)\nu(x)  |f(x)|dx+C\epsilon \int s(x)\nu(x) (Q(x,x)|f(x)|+Q(x,x') |f(x')|)dx\\
&\leq C\epsilon \| s \|_\nu \|f\|_\nu+C\epsilon \| s \|_\nu \|Q |f|\|_\nu\\
&\leq 2C\epsilon \| s \|_\nu \|f\|_\nu. 
\end{align*}
Taking $\mu(x)\propto |(Q-\Qhat)f(x)|\nu(x)$, we have the result.
\end{proof}

\begin{proof}[Proof of Proposition \ref{prop:re}] 
Recall that
\[
P=MQ,\quad M=\left(M_{0}\otimes \cdots  \otimes M_K\right),\quad Q= \left(\frac{1}{K}\sum_{k\in \{0,\ldots,K-1\}} Q_{k,k+1}\right).
\]
and
\[
\Phat=\Mhat\Qhat,\quad \Mhat=\left(\Mhat_{0}\otimes \cdots  \otimes \Mhat_K\right),\quad \Qhat=\left(\frac{1}{K}\sum_{k\in \{0,\ldots,K-1\}} \Qhat_{k,k+1}\right).
\]
Since $M$ is a product of $M_k$,  Lemma \ref{lem:l2comp} claim 2) indicates that $\|M-\Mhat\|_\Pi\leq C_1\epsilon$ for some $C_1$.  
Then note that if $a\leq C A, b\leq CB$ then $\min \{a,b\}\leq C\min\{A,B\}$ so 
the acceptance probability of $Q_{k,k+1}$ and $\Qhat_{k,k+1}$ satisfies 
\[
\frac{\alpha_k(x,x')}{\alphahat_k(x,x')}\leq \sup_{x,x'}\left\{\frac{\pi_k(x')\pi_{k+1}(x)\pihat_{k}(x)\pihat_{k+1}(x')}{\pihat_k(x')\pihat_{k+1}(x)\pi_{k}(x)\pi_{k+1}(x')}\right\}
\leq (1+C_1\epsilon)^4\leq 1+D\epsilon
\]
for some constant $D$. Then Lemma \ref{lem:REl2} indicates that $\|Q_{k,k+1}-\Qhat_{k,k+1}\|_\Pi\leq C_2\epsilon$ for some $C_2$. Then Lemma \ref{lem:l2comp} claim 3) indicates that for some $C_3$
\[
\|Q-\Qhat\|_\Pi\leq C_3\epsilon. 
\] 
Finally, we use claim 1) from Lemma \ref{lem:l2comp} and find that $\|P-\Phat\|_\Pi\leq C'\epsilon$ for some $C'$. 
\end{proof}

\bibliographystyle{plain}
\bibliography{Ref}

\begin{thebibliography}{10}

\bibitem{andrieu2003introduction}
Christophe Andrieu, Nando De~Freitas, Arnaud Doucet, and Michael~I Jordan.
\newblock An introduction to mcmc for machine learning.
\newblock {\em Machine learning}, 50(1):5--43, 2003.

\bibitem{bakry2014analysis}
Dominique Bakry, Ivan Gentil, Michel Ledoux, et~al.
\newblock {\em Analysis and geometry of Markov diffusion operators}, volume
  103.
\newblock Springer, 2014.

\bibitem{bardenet2014adaptive}
R{\'e}mi Bardenet, Arnaud Doucet, and Chris Holmes.
\newblock An adaptive subsampling approach for mcmc inference in large
  datasets.
\newblock In {\em Proceedings of The 31st International Conference on Machine
  Learning}, 2014.

\bibitem{beskos2018multilevel}
Alexandros Beskos, Ajay Jasra, Kody Law, Youssef Marzouk, and Yan Zhou.
\newblock Multilevel sequential monte carlo with dimension-independent
  likelihood-informed proposals.
\newblock {\em SIAM/ASA Journal on Uncertainty Quantification}, 6(2):762--786,
  2018.

\bibitem{bui2013computational}
Tan Bui-Thanh, Omar Ghattas, James Martin, and Georg Stadler.
\newblock A computational framework for infinite-dimensional bayesian inverse
  problems part i: The linearized case, with application to global seismic
  inversion.
\newblock {\em SIAM Journal on Scientific Computing}, 35(6):A2494--A2523, 2013.

\bibitem{dong2020gap}
Jing Dong and Xin~T Tong.
\newblock Spectral gap of replica exchange langevin diffusion on mixture
  distributions.
\newblock {\em arXiv preprint arXiv:2006.16193}, 2020.

\bibitem{dong2021replica}
Jing Dong and Xin~T Tong.
\newblock Replica exchange for non-convex optimization.
\newblock {\em Journal of Machine Learning Research}, 22(173):1--59, 2021.

\bibitem{dupuis2012infinite}
Paul Dupuis, Yufei Liu, Nuria Plattner, and Jimmie~D Doll.
\newblock On the infinite swapping limit for parallel tempering.
\newblock {\em Multiscale Modeling \& Simulation}, 10(3):986--1022, 2012.

\bibitem{durmus2017nonasymptotic}
Alain Durmus and Eric Moulines.
\newblock Nonasymptotic convergence analysis for the unadjusted {L}angevin
  algorithm.
\newblock {\em The Annals of Applied Probability}, 27(3):1551--1587, 2017.

\bibitem{dwivedi2018log}
Raaz Dwivedi, Yuansi Chen, Martin~J Wainwright, and Bin Yu.
\newblock Log-concave sampling: Metropolis-hastings algorithms are fast!
\newblock In {\em Conference on learning theory}, pages 793--797. PMLR, 2018.

\bibitem{JMLR:v20:19-306}
Raaz Dwivedi, Yuansi Chen, Martin~J. Wainwright, and Bin Yu.
\newblock Log-concave sampling: Metropolis-hastings algorithms are fast.
\newblock {\em Journal of Machine Learning Research}, 20(183):1--42, 2019.

\bibitem{earl2005}
David~J Earl and Michael~W Deem.
\newblock Parallel tempering: Theory, applications, and new perspectives.
\newblock {\em Physical Chemistry Chemical Physics}, 7(23):3910--3916, 2005.

\bibitem{hairer2014spectral}
Martin Hairer, Andrew~M Stuart, and Sebastian~J Vollmer.
\newblock Spectral gaps for a {Metropolis--Hastings} algorithm in infinite
  dimensions.
\newblock {\em The Annals of Applied Probability}, 24(6):2455--2490, 2014.

\bibitem{herve2014approximating}
Lo{\"\i}c Herv{\'e} and James Ledoux.
\newblock Approximating markov chains and v-geometric ergodicity via weak
  perturbation theory.
\newblock {\em Stochastic Processes and their Applications}, 124(1):613--638,
  2014.

\bibitem{jones2004markov}
Galin~L Jones.
\newblock On the markov chain central limit theorem.
\newblock {\em Probability surveys}, 1:299--320, 2004.

\bibitem{joulin2010curvature}
Ald{\'e}ric Joulin and Yann Ollivier.
\newblock Curvature, concentration and error estimates for markov chain monte
  carlo.
\newblock {\em The Annals of Probability}, 38(6):2418--2442, 2010.

\bibitem{lindvall02}
Torgny Lindvall.
\newblock {\em Lectures on the coupling method}.
\newblock Courier Corporation, 2002.

\bibitem{lotka1925elements}
Alfred~James Lotka.
\newblock {\em Elements of physical biology}.
\newblock Williams \& Wilkins, 1925.

\bibitem{medina2020perturbation}
Felipe Medina-Aguayo, Daniel Rudolf, and Nikolaus Schweizer.
\newblock Perturbation bounds for monte carlo within metropolis via restricted
  approximations.
\newblock {\em Stochastic processes and their applications}, 130(4):2200--2227,
  2020.

\bibitem{meyn2012markov}
Sean~P Meyn and Richard~L Tweedie.
\newblock {\em Markov chains and stochastic stability}.
\newblock Springer Science \& Business Media, 2012.

\bibitem{parno2018transport}
Matthew~D Parno and Youssef~M Marzouk.
\newblock Transport map accelerated markov chain monte carlo.
\newblock {\em SIAM/ASA Journal on Uncertainty Quantification}, 6(2):645--682,
  2018.

\bibitem{roberts1996exponential}
Gareth~O Roberts and Richard~L Tweedie.
\newblock Exponential convergence of langevin distributions and their discrete
  approximations.
\newblock {\em Bernoulli}, 2(4):341--363, 1996.

\bibitem{rudolf2018perturbation}
Daniel Rudolf and Nikolaus Schweizer.
\newblock Perturbation theory for markov chains via wasserstein distance.
\newblock {\em Bernoulli}, 24(4A):2610--2639, 2018.

\bibitem{shardlow2000perturbation}
Tony Shardlow and Andrew~M Stuart.
\newblock A perturbation theory for ergodic markov chains and application to
  numerical approximations.
\newblock {\em SIAM journal on numerical analysis}, 37(4):1120--1137, 2000.

\bibitem{stuart2010inverse}
Andrew~M Stuart.
\newblock Inverse problems: a bayesian perspective.
\newblock {\em Acta numerica}, 19:451--559, 2010.

\bibitem{sugita1999replica}
Yuji Sugita and Yuko Okamoto.
\newblock Replica-exchange molecular dynamics method for protein folding.
\newblock {\em Chemical physics letters}, 314(1-2):141--151, 1999.

\bibitem{tawn2019accelerating}
Nicholas~G Tawn and Gareth~O Roberts.
\newblock Accelerating parallel tempering: Quantile tempering algorithm
  (quanta).
\newblock {\em Advances in Applied Probability}, 51(3):802--834, 2019.

\bibitem{tawn2020weight}
Nicholas~G Tawn, Gareth~O Roberts, and Jeffrey~S Rosenthal.
\newblock Weight-preserving simulated tempering.
\newblock {\em Statistics and Computing}, 30(1):27--41, 2020.

\bibitem{tierney1994markov}
Luke Tierney.
\newblock Markov chains for exploring posterior distributions.
\newblock {\em the Annals of Statistics}, pages 1701--1728, 1994.

\bibitem{vempala2019rapid}
Santosh Vempala and Andre Wibisono.
\newblock Rapid convergence of the unadjusted langevin algorithm: Isoperimetry
  suffices.
\newblock {\em Advances in neural information processing systems}, 32, 2019.

\bibitem{woodard2009conditions}
Dawn~B Woodard, Scott~C Schmidler, and Mark Huber.
\newblock Conditions for rapid mixing of parallel and simulated tempering on
  multimodal distributions.
\newblock {\em The Annals of Applied Probability}, 19(2):617--640, 2009.

\end{thebibliography}
\end{document}